\documentclass[11pt]{article}

\usepackage{amsmath,amssymb,amsthm,mathtools}

    \usepackage{thmtools}
    \usepackage{thm-restate}
	\usepackage{a4,geometry}

\usepackage{graphicx}
\usepackage{paralist}
\usepackage{bm}
\usepackage{xspace}
\usepackage{url}
\usepackage{fullpage, prettyref}
\usepackage{boxedminipage}
\usepackage{wrapfig}
\usepackage{ifthen}
\usepackage{color}
\usepackage[dvipsnames]{xcolor}
\usepackage[colorlinks,citecolor=blue,linkcolor=BrickRed]{hyperref}
\usepackage{framed}
\usepackage{cleveref}
\usepackage{algpseudocode}
\usepackage[ruled,vlined, linesnumbered]{algorithm2e}
\usepackage{titlesec}
\titlespacing*{\subsubsection}{0pt}{0.2em}{0.2em}

\usepackage{thmtools}
\usepackage{thm-restate}
\usepackage{cleveref}

\makeatletter
\@ifundefined{newcounteralias}{}{%
  \renewcommand\thmt@autorefsetup{%
    \@xa\def\csname\thmt@envname autorefname\@xa\endcsname
      \@xa{\thmt@thmname}%
  }%
}
\makeatother

\newtheorem{theorem}{Theorem}[section]

\newtheorem{lemma}[theorem]{Lemma}
\newtheorem{claim}[theorem]{Claim}

\newtheorem{remark}[theorem]{Remark}

\newcommand{\ignore}[1]{}

\newcommand{\R}{\mathbb R}

\newcommand{\eps}{\varepsilon}

\newcommand{\poly}{\mathsf{poly}}

\newcommand{\KS}{\mathsf{KS}}

\newcommand{\diag}{\mathop{\mathsf{diag}}}
\newcommand{\Cov}{\mathsf{Cov}}

\newcommand{\norm}[1]{\left\lVert #1 \right\rVert}

\newcommand{\E}{\mathbb{E}}

\newcommand\restr[2]{{
  \left.\kern-\nulldelimiterspace
  #1 
  \vphantom{\big|} 
  \right|_{#2} 
 }}

\newcommand{\Sec}[1]{\hyperref[sec:#1]{\S\ref*{sec:#1}}} 
\newcommand{\Eqn}[1]{\hyperref[eq:#1]{(\ref*{eq:#1})}} 
\newcommand{\Fig}[1]{\hyperref[fig:#1]{Fig.\,\ref*{fig:#1}}}
\newcommand{\Tab}[1]{\hyperref[tab:#1]{Tab.\,\ref*{tab:#1}}}
\newcommand{\Thm}[1]{\hyperref[thm:#1]{Theorem\,\ref*{thm:#1}}} 
\newcommand{\Fact}[1]{\hyperref[fact:#1]{Fact\,\ref*{fact:#1}}} 
\newcommand{\Lem}[1]{\hyperref[lem:#1]{Lemma~\ref*{lem:#1}}} 
\newcommand{\Prop}[1]{\hyperref[prop:#1]{Proposition~\ref*{prop:#1}}} 
\newcommand{\Cor}[1]{\hyperref[cor:#1]{Corollary~\ref*{cor:#1}}} 
\newcommand{\Conj}[1]{\hyperref[conj:#1]{Conjecture~\ref*{conj:#1}}} 
\newcommand{\Rem}[1]{\hyperref[rem:#1]{Remark~\ref*{rem:#1}}} 
\newcommand{\Def}[1]{\hyperref[def:#1]{Definition~\ref*{def:#1}}} 
\newcommand{\Alg}[1]{\hyperref[alg:#1]{Alg.~\ref*{alg:#1}}} 
\newcommand{\Ex}[1]{\hyperref[ex:#1]{Ex.~\ref*{ex:#1}}} 
\newcommand{\Clm}[1]{\hyperref[clm:#1]{Claim~\ref*{clm:#1}}} 
\newcommand{\Step}[1]{\hyperref[step:#1]{Step~\ref*{step:#1}}} 
\newcommand{\Obs}[1]{\hyperref[obs:#1]{Observation~\ref*{obs:#1}}} 
\renewcommand{\Alg}[1]{\hyperref[alg:#1]{Algorithm~\ref*{alg:#1}}} 

\graphicspath{{./Figures/}}

\newcommand{\Tr}{\operatorname{Tr}}

\RequirePackage[normalem]{ulem}

\title{Kadison--Singer in Deterministic Polynomial Time}

\date{}

\author{Ekene Ezeunala\thanks{University of Chicago, Chicago, IL, USA. \texttt{ekene@uchicago.edu.}}
\and
Haotian Jiang\thanks{University of Chicago, Chicago, IL, USA. \texttt{jhtdavid@uchicago.edu}.}}

\RequirePackage[normalem]{ulem}

\begin{document}

\allowdisplaybreaks
\begin{titlepage}
\maketitle

\begin{abstract}
We give a deterministic polynomial-time algorithm that, given rational Hermitian matrices $H_1,\dots,H_N$ of rank at most one, finds signs $s\in\{\pm1\}^N$ with $\|\sum_i s_i H_i\|\le 13\|\sum_i H_i^2\|^{1/2}$. As a corollary, for vectors $v_i$ with $\sum_i v_iv_i^*=I$ and $\|v_i\|^2\le\delta$, the signs yield a partition $[N] = S_1 \cup S_2$ such that each part satisfies $\|\sum_{i \in S_j} v_i v_i^* - \frac{I}{2}\| \leq \frac{13}{2}\sqrt\delta$ for $j = 1,2$. This gives a deterministic polynomial-time algorithm for the Kadison--Singer problem, in Weaver's equivalent discrepancy-theoretic $\mathsf{KS}_2$ formulation, with a universal constant. 
\end{abstract}

\vspace{1cm}

{\hypersetup{linkcolor=BrickRed}
 \tableofcontents
}
 \thispagestyle{empty}
\end{titlepage}

\thispagestyle{empty}
\newpage
\setcounter{page}{1}

\section{Introduction}Let $v_1,\dots,v_N\in\mathbb C^d$ satisfy $\sum_iv_iv_i^*=I$ and $\|v_i\|^2\le\delta$ for every $i$. The Kadison--Singer problem 
\cite{KS59}, in the equivalent discrepancy-theoretic $\KS_2$ formulation of Weaver 
\cite{Wea04}, asks for a partition of $[N]$ into two parts $S_1 \cup S_2$ such that each one of the parts $\sum_{i\in S_j} v_iv_i^*$ is close to $\frac{1}{2}I$ within an error that depends only on $\delta$ for $j = 1,2$.

In breakthrough work, Marcus, Spielman, and Srivastava \cite{MSS15b} 
proved that a partition with $\smash{\|\sum_{i\in S_j}v_iv_i^*\|\le\frac{1}{2}(1+\sqrt{2\delta})^2}$ for both parts always exists. Their proof uses interlacing families of polynomials and does not yield an efficient algorithm. Jourdan, Macgregor, and Sun~\mbox{
\cite{JMS22} }\hskip0pt
formulated the algorithmic version as the problem $\KS_2(c)$, which asks for a polynomial time algorithm to find a partition with $\smash{\|\sum_{i\in S_j}v_iv_i^*-I/2\|\le c\sqrt\delta}$ for constant $c$.
The previous best algorithm for this problem, due to  Anari, Oveis Gharan, Saberi, and Srivastava \cite{AOSS18}, runs in time $2^{\widetilde{O}(n^{1/3})}$.

In this article, we give a deterministic polynomial-time algorithm for this problem with a universal constant. It is a special case of the following algorithmic rank-one matrix discrepancy result.
A stronger bound with the constant $4$ in the setting of \Cref{thm:signing} was shown in \cite{KLS20}, but their proof is also based on interlacing polynomials and does not yield an efficient algorithm.  

\begin{theorem}[Algorithmic rank-one matrix discrepancy]\label{thm:signing}
There is a deterministic polynomial-time algorithm that, given rational Hermitian matrices $H_1,\dots,H_N\in\mathbb C^{d\times d}$ of rank at most one, outputs signs $s\in\{\pm1\}^N$ satisfying \[
\Big\|\sum_{i=1}^Ns_iH_i \Big\|\le13 \Big\|\sum_{i=1}^NH_i^2 \Big\|^{1/2}.
\]
\end{theorem}

For the $\KS_2$ problem above,  the identity $(v_iv_i^*)^2=\|v_i\|^2v_iv_i^*$ gives $\sum_i(v_iv_i^*)^2\preceq\delta I$. Applying \Cref{thm:signing} to $H_i=v_iv_i^*$, the sign class $S_1 :=\{i:s_i=1\}$ and $S_2 := \{i:s_i=-1\}$ satisfies
\[ \Big\|\sum_{i\in S_j}v_iv_i^*-\frac{1}{2}I \Big\|
= \frac{1}{2} \Big\|\sum_is_iv_iv_i^* \Big\|\le \frac{13}{2}\sqrt{\delta},
\]
where $j = 1, 2$.
Thus the algorithm solves $\KS_2(13/2)$ in the notation of~\cite{JMS22}. We have not attempted to optimize the constant to keep our presentation clean, and the existence bound of  \cite{MSS15b} remains stronger.

\smallskip
\noindent \textbf{Overview.}
We may assume via a normalization (see \eqref{eq:normalization}) that the input consists of PSD rank-one matrices $A_i$ with $\Tr\sum_iA_i=1$, and the task is to control $S(x)=\sum_ix_iA_i$ along a walk of a fractional signing $x\in[-1,1]^N$ from $x=0$ to a vertex, guided by a potential function. While potential-based algorithms have been used extensively in discrepancy theory (e.g., \cite{BLV22,PV23,BJ25b,AS26}), the potential function and the steps used here are substantially more involved.

The potential function is composed of two parts: (1) a matrix part $R(x)$, specified below, and (2) a coordinate part where each active coordinate carries a budget $\psi_i(x) =(1-x_i^2)^{1/3}$, which vanishes once the coordinate is frozen at $\pm1$. The matrix part $R(x)$ is the optimal value of a semidefinite program in $t\in\R$  and matrix variables $X,Y\succ0$ whose two constraints,
\[
X^{-1}+S(x)+c\sum_i\psi_i(x) A_iYA_i\preceq tI,\qquad
Y^{-1}-S(x)+c\sum_i\psi_i(x) A_iXA_i\preceq tI,\]
control the spectrum of $S(x)$. In particular, $R(x)\ge\|S(x)\|$ and $R(0)=O(\|\sum_iA_i^2\|^{1/2})$. The terms involving $\psi_i(x)$ act as reservoirs; each holds a PSD budget along $A_i$, weighted by the resolvent of the opposite constraint, that is released when $x_i$ moves. The potential function used by the algorithm is of the form $\Psi(x)=R(x)+\lambda\sum_i\psi_i(x)$. At each step, the algorithm finds a deterministic move that strictly decreases $\Psi(x)$ by at least $1/\poly(N)$, so the final discrepancy is controlled by the initial potential and the number of steps is polynomially bounded. 

The main step of the analysis (\Cref{lem:potential_decrease}) shows that every state $x \in [-1,1]^N$ with an active coordinate admits a decreasing move. If some reservoir, i.e., $\psi_i(x) A_i Y A_i$ or $\psi_i(x) A_i X A_i$ in $R(x)$, is {\em heavy} enough to absorb the jump of an active coordinate $i$ to an endpoint $\pm 1$, that jump keeps the current $(t,X,Y)$ feasible and removes $\lambda\psi_i(x)$ from $\Psi(x)$, and hence reduces the potential.

Otherwise every active coordinate is {\em light} in both constraints. We then express, using rank-one identities, the change of $X$ and $Y$ due to a local step as a coupled linear system on the active coordinates (see \eqref{eq:reformed}) with the following tension: the convexity of matrix inversion $X^{-1}$ forces an increase in $R(x)$; while the concavity of $\psi_i(x)$ in the reservoirs reduces $R(x)$. 
The covariance lemma of \cite{Ban24}, applied on a subspace chosen to control the curvature of matrix inversion, gives a drift and a covariance (for a random step) under which the concavity of the reservoirs outweighs the cost of inversion, resulting in a net decrease in $R(x)$ (see \Cref{lem:potential_decrease}). 

To give a deterministic algorithm, we will only use the existence of a good random step as a certificate: when the gradient is small, it forces a negative eigenvalue of the Hessian. 
The algorithm therefore compares finitely many candidates at each step, namely the endpoint moves of active coordinates, a coordinate step against the gradient, and a step along a minimum-eigenvalue direction of the Hessian, and takes the best one. The random step certificate ensures that every iteration decreases the potential by an inverse-polynomial amount, which bounds the number of iterations. 

\smallskip
\noindent \textbf{Concurrent and Independent Work.} In concurrent and independent work, Kathuria \cite{Kat26} and Song and Yue \cite{SY26} also gave polynomial time algorithms for Weaver's discrepancy formulation of the Kadison-Singer problem. The techniques used are similar; in particular, the potential function used in \cite{SY26} is the same as the one in this article.   
\section{Preliminaries}\label{sec:prelims}
We write $\norm{\cdot}$ for the operator norm of a matrix or the Euclidean norm of a vector, $\norm{\cdot}_{\mathrm F}$ for the Frobenius norm, and $\langle U,V\rangle=\Tr(UV)$ on Hermitian matrices. 
Matrix calculations take place over the real vector space of Hermitian matrices.

For $A=vv^*\succeq0$ of rank one, we have 
\begin{equation}\label{eq:rank-one}
 A^2=\Tr(A)A,\qquad AZA=\Tr(AZ)A,\qquad
 |v_i^*Xv_j|^2\le(v_i^*Xv_i)(v_j^*Xv_j)\quad(\text{for }X\succ0).
\end{equation}
We use matrix convexity of inversion and the Schur-complement equivalence
\begin{align} \label{eq:schur-complement-equiv}
 X^{-1}\preceq D\quad\Longleftrightarrow\quad
 \begin{pmatrix}D&I\\I&X\end{pmatrix}\succeq0\qquad(\text{for }X\succ0).
\end{align}
We will also use the following covariance lemma, which is a form of 
\cite[Theorem 2.5]{Ban24}. For a matrix $U$, we denote by $\diag(U)$ its diagonal part.
\begin{lemma}\label{lem:covariance}
If $L\subseteq\R^m$ has dimension $\ell$ and $\kappa>1$, there is a matrix $U\succeq0$ such that
\[
 \operatorname{range}(U)\subseteq L,\qquad U_{ii}\le1 \text{ for all $i \in [m]$},\qquad
 U\preceq\kappa\diag(U),\qquad \Tr U\ge\ell-m/\kappa.
\]
\end{lemma}

\smallskip
\noindent \textbf{Input Normalization.} Let $H_i$ be the input matrices to \Cref{thm:signing}, which we may assume are all non-zero. We renormalize the matrices $H_i$ to PSD matrices $A_i$ as 
\begin{align} \label{eq:normalization}
A_i = \frac{\mathsf{sign}(\operatorname{Tr} H_i)H_i}{\sum_i |\operatorname{Tr} H_i|} ,
\end{align}
which satisfy $\Tr\sum_iA_i=1$. Set $\nu := \|\sum_i A_i^2\|$ (which we assume to be rational\footnote{The rationality assumption ensures that the algorithm has polynomially-bounded bit complexity. If $\nu$ is not rational, one may replace $\nu$ with any rational number between $[\nu, 1.001\nu]$.}). Then,
\begin{align*}
 1 \geq \nu \geq \big\|\sum_i A_i\big\|^2/N \geq \frac{1}{Nd^2} ,
\end{align*}
where the second inequality follows by expanding $\sum_i (A_i - (\sum_i A_i)/N)^2 \succeq 0$.

\section{Proof of Main Result}\label{sec:proof_main}
The algorithm in \Cref{thm:signing} rounds a fractional coloring $x \in [-1,1]^N$ to a $\pm 1$-signing. Throughout the rest of this article, we work with the normalized PSD matrices $A_1,\ldots,A_N$ from \eqref{eq:normalization}. We allow an arbitrary rational starting point $x^0 \in [-1,1]^N$ and write
\[
S(x) = \sum_{i=1}^N (x_i-x_i^0)A_i.
\]
Our goal is to reach a vertex with $\|S(x)\| = O(\sqrt{\nu})$; taking $x^0=0$ gives \Cref{thm:signing}.

At each state $x \in [-1,1]^N$ (we will not include a subscript with time $t$ to simplify the notation), a coordinate is \emph{active} if $|x_i|<1$, and is {\em frozen} once it reaches an endpoint $\pm 1$. A state $x \in [-1,1]^N$ is called {\em non-terminal} if it has at least one active coordinate. All derivatives below are taken with respect to the active coordinates, with the frozen coordinates fixed.

\subsection{Potential Function and Key Lemmas} 
\label{subsec:potential}

We measure the remaining fractionality of the current coloring $x \in [-1,1]^N$ using potential 
\begin{align} \label{eq:coordinate_potential}
\psi_i = \psi_i(x) = (1-x_i^2)^{1/3},
\qquad
\Phi(x) = \sum_i \psi_i(x) .
\end{align}
Clearly $\Phi(x) \ge 0$, with equality exactly at the vertices.

The same weights enter the matrix part of the potential through the positive linear map
\[
\eta_x(Z) = c\sum_i \psi_i A_i Z A_i, 
\]
where $c := 40$ is a large constant. 
Setting $\varepsilon := \nu/d$, the matrix potential is defined as 
\begin{equation} \label{eq:matrix_potential}
\begin{aligned}
R(x) = \min_{t \in \mathbb{R},\, X,Y \succ 0}\quad
& t + \varepsilon\operatorname{Tr}(X+Y) \\
\text{subject to}\quad
& X^{-1} + S(x) + \eta_x(Y) \preceq tI, \\
& Y^{-1} - S(x) + \eta_x(X) \preceq tI .
\end{aligned}
\end{equation}
Here, the extra term $\varepsilon \Tr(X+Y)$ is crucial for analytic reasons: it is used to ensure that the optimal value is uniquely attained and the optimizer is smooth with $x$ in each open face of the cube.
Note that via the Schur-complement equivalence in \eqref{eq:schur-complement-equiv}, the program \eqref{eq:matrix_potential} defining $R(x)$ can be written as a semidefinite program in $(t,X,Y)$, which is how the algorithm evaluates $R(x)$ efficiently.\footnote{For the analysis it is more convenient to work with the program formulation in  \eqref{eq:matrix_potential}, as the Lagrangian multipliers are $d\times d$ matrices. See \Cref{sec:proof_key_lemmas} for details. }

The final potential used by the algorithm will be a weighted combination of \eqref{eq:coordinate_potential} and \eqref{eq:matrix_potential}: 
\begin{align} \label{eq:total_potential}
    \Psi(x) := R(x) + \lambda\Phi(x) ,
\end{align}
where $\lambda := \sqrt{\nu}/100N$ is assumed to be rational (if not, we will replace it with any rational number between $[\sqrt{\nu}/200N, \sqrt{\nu}/100N]$).

Note that the two constraints in the matrix potential \eqref{eq:matrix_potential} bound the positive and negative eigenvalues of $S(x)$, while also accounting for the unfinished coordinates through $\eta_x$; rounding a coordinate changes $S(x)$, but removes that coordinate's contribution to $\eta_x$. This is the exchange that drives the proof. The second term $\lambda\Phi(x)$ is useful for a few reasons: it ensures that freezing a coordinate or a mean-zero random local step gives a strict decrease in $\Psi(x)$ whenever $R(x)$ does not increase, which controls the number of such steps; it additionally allows one to round coordinates $x_i$ that are very close to an endpoint to that endpoint without increasing the potential function.

\smallskip
\noindent \textbf{Two Key Lemmas.} We first state two key lemmas and use them to prove \Cref{thm:signing} in \Cref{subsec:proof_main_thm}. The proof of these two key lemmas will be deferred to \Cref{sec:proof_key_lemmas}.

The first lemma shows how the matrix potential controls the spectrum of $S(x)$. 

\begin{restatable}[Matrix potential]{lemma}{MatrixPotential}\label{lem:matrix_potential}
For every $x \in [-1,1]^N$, the optimization problem defining $R(x)$ has a unique minimizer. The minimizer and $R(x)$ depend smoothly on $x$ within each open face of the cube, with the frozen coordinates fixed. Moreover, we have the bound
\[
\begin{aligned}
\|S(x)\| &\le R(x) \le \|S(x)\| + 2\sqrt{c\nu+2d\varepsilon} = \|S(x)\| + 2\sqrt{(c+2)\nu}, \\
\Psi(x^0) &\le \left(2\sqrt{42}+\frac{1}{100}\right)\sqrt{\nu} \leq 13 \sqrt{\nu}.
\end{aligned}
\]
\end{restatable}
In particular, reaching a vertex without increasing $\Psi$ gives the required discrepancy bound. 

The second key lemma supplies progress at every nonterminal state $x \in [-1,1]^N$.

\begin{restatable}[Decreasing the potential]{lemma}{PotentialDecrease} \label{lem:potential_decrease}
There is a $C_0 \leq \poly(N, d)$ independent of the non-terminal state $x$ such that the following holds. Let $n \ge 1$ be the number coordinates that are active at state $x$. Then either assigning an active coordinate $i$ to an endpoint decreases $\Psi(x)$ by at least $\lambda\psi_i(x)$, or there are $\mu \in \mathbb{R}^n$ and $\Sigma \succeq 0$ satisfying
\[
\|\mu\|_\infty \le C_0,
\qquad
\operatorname{Tr}\Sigma \le C_0n,
\]
such that a random step with drift $\mu$ and covariance $\Sigma$ decreases the potential as 
\[
\nabla\Psi(x)\cdot\mu
+ \frac{1}{2}\operatorname{Tr}\bigl(\nabla^2\Psi(x)\Sigma\bigr)
\le -\frac{\lambda n}{C_0}.
\]
\end{restatable}
The endpoint alternative in \Cref{lem:potential_decrease} holds when removing one coordinate's contribution to $\eta_x$ in either constraint of \eqref{eq:matrix_potential} compensates for its entire change to $S(x)$. 
If this is not the case for any active coordinate $i$, then we consider the random step alternative. 

Although the second alternative in \Cref{lem:potential_decrease} is stated as a random step, it yields a deterministic descent step as follows: a sufficiently large gradient gives a coordinate direction of descent; otherwise, the displayed inequality forces a negative eigenvalue of the Hessian. Our algorithm  will never compute the $\mu$ and $\Sigma$ in \Cref{lem:potential_decrease}: they only serve as a certificate that a potential-decreasing step exists. We formalize this intuition in the proof of \Cref{thm:signing} below.

\subsection{Proof of \texorpdfstring{\Cref{thm:signing}}{signing}}
\label{subsec:proof_main_thm}

Now we can prove the main result.
\begin{proof}[Proof of \Cref{{thm:signing}}]
We use \Cref{lem:potential_decrease} to round $x^0$ to a vertex while decreasing $\Psi$ at each iteration. To achieve this, we first round coordinates that are sufficiently close to an endpoint.

Set $\sigma := \lambda^2/\nu$. Suppose an active coordinate $x_i$ has distance $z \le \sigma$ from its nearer $\pm1$-endpoint. Moving it to that endpoint removes its contribution to $\eta_x$ and hence decreases it. Since $\|A_i\| \le \sqrt{\nu}$, keeping $X,Y$ and increasing $t$ by $z\sqrt{\nu}$ preserves both constraints in \eqref{eq:matrix_potential}. Hence
\[
\Delta\Psi(x)
\le z\sqrt{\nu} - \lambda (z(2-z))^{1/3}
\le z\sqrt{\nu} - \lambda\sqrt{z}
\le 0,
\]
where the second inequality uses $z^{1/3} \geq z^{1/2}$ for every $z \in [0,1]$, and the last inequality follows from $\sqrt{z\nu} \le \sqrt{\sigma\nu} = \lambda$. We perform this ``rounding to nearer endpoint'' procedure at the start of every iteration, freezing each coordinate that reaches an endpoint. Every remaining active coordinate then has distance greater than $\sigma$ from both endpoints.

\smallskip
\noindent \textbf{The Three Types of Movements.}
Let $M \leq \poly(N)$ be an upper bound on the operator norm of the second and third derivatives of $\Psi(x)$ whenever all active coordinates have distance at least $\sigma$ from the endpoints. 
Such a bound exists by \Cref{lem:matrix_potential} and compactness on each of the finitely many faces. 
Set step size parameters (which are both at least $1/\poly(N)$) 
\[
a := \lambda/(8C_0^2),
\qquad
s := \min\big\{\sigma,a / M\big\}.
\]
At each iteration where the current state $x \in [-1,1]^N$ has $n \ge 1$ active coordinates (after ``rounding to nearer endpoint'' above), consider the following three types of movements: 
\begin{enumerate}
    \item [(1)] {\it Endpoint movement:} The $2n$ single-coordinate endpoint movements, where an active coordinate $x_i$ is moved to either $\pm 1$.
    \item [(2)] {\it Hessian descent movement:} The two possible movements $x \pm sv$, where $v$ is a unit eigenvector corresponding to the smallest eigenvalue of the Hessian $\nabla^2 \Psi(x)$. 
    \item [(3)] {\it Gradient coordinate movement:} For an active coordinate $i$ with the largest gradient entry magnitude, i.e., $i \in \arg\max_j |\partial_j\Psi(x)|$, the coordinate step $x - s\,\mathsf{sign}\bigl(\partial_i\Psi(x)\bigr)\,\mathbf{e}_i$. 
\end{enumerate}
The algorithm tries all three types of movements, and eventually takes the step that attains the smallest potential at the current iteration.

\smallskip
\noindent \textbf{Potential Decrease.} 
We claim that whenever the current state $x \in [-1,1]^N$ is not a terminal state, i.e., some active coordinate $x_i$ is at least $\sigma$-far from an endpoint, one of the three movements above decreases $\Psi$ by at least $as^2 \geq 1/\poly(N)$.  

If the first alternative of \Cref{lem:potential_decrease} holds, an endpoint movement decreases the potential by 
\[
\lambda\psi_i(x) \ge \lambda\sigma^{1/3} \ge as^2.
\]
If $\|\nabla\Psi(x)\|_\infty \ge a$, the gradient coordinate movement  gives
\[
\Delta\Psi(x)
\le -as + \frac{1}{2}Ms^2
\le -\frac{1}{2}as
\le -as^2.
\]
It remains to consider a state $x \in [-1,1]^N$ with $\|\nabla\Psi(x)\|_\infty < a$ at which the second alternative of \Cref{lem:potential_decrease} holds. 
Note that $|\nabla \Psi(x) \cdot \mu| \leq a C_0 n$, hence the guarantee of \Cref{lem:potential_decrease} implies 
\[
\frac{1}{2}\operatorname{Tr}(\nabla^2\Psi\,\Sigma)
\le -\frac{\lambda n}{C_0} + aC_0n
\le -\frac{\lambda n}{2C_0}.
\]
Since $0 < \operatorname{Tr}\Sigma \le C_0n$, this further implies that for some unit vector $v$, 
\[
v^\top \nabla^2 \Psi(x) v = \lambda_{\min}(\nabla^2\Psi(x))
\le -\frac{\lambda}{C_0^2}
= -8a.
\]
Averaging the Taylor expansions along $v$ and $-v$ cancels the linear term:
\[
\begin{aligned}
\frac{\Psi(x+sv)+\Psi(x-sv)}{2}
&\le \Psi(x) + \frac{s^2}{2}v^\top\nabla^2\Psi\,v
    + \frac{1}{6}Ms^3 \\
&\le \Psi(x) - 4as^2 + \frac{1}{6}Ms^3 
\le \Psi(x) - 3as^2.
\end{aligned}
\]
Thus one of the two candidates $x \pm v$ has the required decrease.

It follows that at any non-terminal state $x \in [-1,1]^N$, one of the three types of movements described above decreases the potential $\Psi(x)$ by at least $1/\poly(N)$, hence the algorithm terminates in at most $\poly(N)$ iterations. Moreover, each iteration of the algorithm can be implemented in polynomial time using semidefinite programming (for evaluating the matrix potential $R(x)$) and eigenvalue computation (for generating the movements). 
Finally, the algorithm reaches a vertex $x^f \in \{\pm 1\}^N$ which, by \Cref{lem:matrix_potential} and the fact that $\Phi(x^f) = 0$, satisfies the bound
\[
\Big\|\sum_i (x_i^f-x_i^0)A_i\Big\|
\le R(x^f)
= \Psi(x^f)
\le \Psi(x^0)
\le \left(2\sqrt{42}+\frac{1}{100}\right)\sqrt{\nu} \leq 13 \sqrt{\nu }.
\]   
This completes the proof of the \Cref{thm:signing}.
\end{proof}

\section{Proof of the Key Lemmas}
\label{sec:proof_key_lemmas}

In this section, we prove the two key lemmas in \Cref{sec:proof_main}, completing the proof. 

\subsection{Proof of \texorpdfstring{\Cref{lem:matrix_potential}}{Lemma 3.1}}

We start with the proof of \Cref{lem:matrix_potential}, which is restated below for convenience. 

\MatrixPotential*

\begin{proof}[Proof of \Cref{lem:matrix_potential}]
Fix $x$ and suppress its dependence in the notation, writing $S=S(x)$ and $\eta=\eta_x$. We use three facts about $\eta$ throughout: (1) it is a positive linear map, i.e., $\eta(Z)\succeq 0$ for every $Z \succeq 0$, hence monotone, i.e., $Z\preceq Z'$ implies $\eta(Z)\preceq\eta(Z')$, (2) it is self-adjoint, $\langle U,\eta(V)\rangle=\langle \eta(U),V\rangle$, by cyclicity of the trace, and (3) $\eta(I)=c\sum_i\psi_iA_i^2\preceq c\nu I$, since $0\le\psi_i\le 1$. We also use that $\|S\|\le 2$, because $|x_i-x_i^0|\le 2$ and $\sum_iA_i\preceq\operatorname{Tr}(\sum_iA_i)\,I=I$. Finally, for $X\succ 0$ and Hermitian $U$, expanding $(X+sU)^{-1}=X^{-1}(I+sUX^{-1})^{-1}$ as a geometric series gives
\begin{equation}\label{eq:inverse-expansion}
(X+sU)^{-1}=X^{-1}-s\,X^{-1}UX^{-1}+s^2\,X^{-1}UX^{-1}UX^{-1}+O(s^3)\qquad(s\to 0),
\end{equation}
which supplies the first two derivatives of matrix inversion.

\smallskip
\noindent\textbf{Bounds on $R$ and $\Psi(x^0)$.}
If $(t,X,Y)$ is feasible, then $X^{-1}\succ 0$ and $\eta(Y)\succeq 0$, so the first constraint gives $S\prec tI$, and the second gives $-S\prec tI$. Hence $t>\|S\|$, and since the objective is at least $t$, we get $R\ge\|S\|$. Conversely, for every $a>0$ the point $X=Y=aI$, $t=\|S\|+a^{-1}+c\nu a$ is feasible, because $a^{-1}I+S+a\,\eta(I)\preceq(a^{-1}+\|S\|+c\nu a)I$ due to the bound $\eta(I) \preceq c \nu I$ above, and likewise for the second constraint. Its objective value is
\[
\|S\|+a^{-1}+(c\nu+2d\varepsilon)a,
\]
and the choice $a=(c\nu+2d\varepsilon)^{-1/2}$ gives $R\le\|S\|+2\sqrt{c\nu+2d\varepsilon}=\|S\|+2\sqrt{(c+2)\nu}$, as $d\varepsilon=\nu$. Finally, $S(x^0)=0$, $\Phi(x^0)\le N$, $c=40$ and $\lambda N=\sqrt{\nu}/100$ give the bound on $\Psi(x^0)$.

\smallskip
\noindent\textbf{Attainment.}
Let $K:=100$, and let $\mathcal C$ be the set of feasible $(t,X,Y)$ with objective value at most $K$. It is nonempty, since $R\le\|S\|+2\sqrt{42\nu}\le 2+2\sqrt{42}<K$. For $(t,X,Y)\in\mathcal C$ we have $t\le K$, because the trace penalty is positive; the first constraint gives $X^{-1}\preceq tI-S-\eta(Y)\preceq(t+\|S\|)I\preceq 2KI$; and $\varepsilon\operatorname{Tr}(X)\le K$. Together with the symmetric bounds for $Y$,
\begin{equation}\label{eq:box}
0<t\le K,
\qquad
(2K)^{-1}I\preceq X,Y\preceq(K/\varepsilon)I
\qquad\text{for every }(t,X,Y)\in\mathcal C.
\end{equation}
Thus $\mathcal C$ is contained in a compact box, and it is closed in this box: the map $X\mapsto X^{-1}$ is continuous on $\{X\succeq(2K)^{-1}I\}$, so the constraints and the sublevel condition are closed conditions there. Hence $\mathcal C$ is compact, and the continuous objective attains its minimum on $\mathcal C$. Feasible points outside $\mathcal C$ have objective value larger than $K$, so this is the minimum of the whole program. The minimum is therefore attained, and every minimizer satisfies \eqref{eq:box}.

\smallskip
\noindent\textbf{Optimality Conditions.}
The program \eqref{eq:matrix_potential} defining $R$ is convex: its objective is linear and, since matrix inversion is operator convex on positive definite matrices, each of its constraints is a convex matrix inequality in $(t,X,Y)$. It is also strictly feasible, for instance at $X=Y=I$ with $t$ sufficiently large. For $P,Q\succeq 0$, define the Lagrangian
\[
\Lambda(t,X,Y;P,Q)
=t+\varepsilon\operatorname{Tr}(X+Y)
+\big\langle P,\;X^{-1}+S+\eta(Y)-tI\big\rangle
+\big\langle Q,\;Y^{-1}-S+\eta(X)-tI\big\rangle .
\]
Note that strong duality holds for this convex program under Slater's condition (see \cite[\S 5.9]{BV04book}).
Thus by the Lagrange multiplier theorem for convex programs with matrix inequality constraints, there exist $P,Q\succeq 0$ such that every minimizer $(t,X,Y)$ of the program \eqref{eq:matrix_potential} 
\begin{enumerate}
\item[(a)] minimizes $\Lambda(\cdot\,;P,Q)$ over all $t\in\mathbb R$ and $X,Y\succ 0$, and
\item[(b)] satisfies complementary slackness:
$\langle P,\,tI-X^{-1}-S-\eta(Y)\rangle=0=\langle Q,\,tI-Y^{-1}+S-\eta(X)\rangle$.
\end{enumerate}
Fix such multipliers and a minimizer $(t,X,Y)$. Since $\Lambda$ is affine in $t$ with slope $1-\operatorname{Tr}(P+Q)$, and $t$ ranges over all of $\mathbb R$, (a) forces this slope to vanish. Next, $\Lambda$ is differentiable on the open set $\{X,Y\succ 0\}$, and by \eqref{eq:inverse-expansion} and the self-adjointness of $\eta$, its gradient with respect to $X$ is $\varepsilon I-X^{-1}PX^{-1}+\eta(Q)$, which vanishes by (a). Conjugating by $X$, and arguing symmetrically for $Y$, we obtain the stationarity equations 
\begin{equation}\label{eq:kkt-stationarity}
\operatorname{Tr}(P+Q)=1,
\qquad
P=\varepsilon X^2+X\eta(Q)X,
\qquad
Q=\varepsilon Y^2+Y\eta(P)Y .
\end{equation}
Since $\eta(Q)\succeq 0$, the second equation gives $P\succeq\varepsilon X^2\succ 0$, and likewise $Q\succ 0$; also $P,Q\preceq I$, because $\operatorname{Tr}(P+Q)=1$. As $P\succ 0$, complementary slackness (b) forces the first constraint to be tight: if $Z\succeq 0$ and $\langle P,Z\rangle=0$, then $\lambda_{\min}(P)\operatorname{Tr}(Z)\le\langle P,Z\rangle=0$, so $Z=0$. The same holds for the second constraint by symmetry, so we have 
\begin{equation}\label{eq:kkt-tight}
X^{-1}+S+\eta(Y)=tI,
\qquad
Y^{-1}-S+\eta(X)=tI .
\end{equation}

\smallskip
\noindent\textbf{Uniqueness of the Minimizer.}
The objective of \eqref{eq:matrix_potential} is linear, so uniqueness has to come from the constraints; the multipliers transfer the strict convexity of matrix inversion to the Lagrangian. Indeed, by \eqref{eq:inverse-expansion}, the second derivative of $X\mapsto\operatorname{Tr}(PX^{-1})$ along a Hermitian direction $U$ is
\[
2\operatorname{Tr}(PX^{-1}UX^{-1}UX^{-1})
=2\operatorname{Tr}\big(\tilde P\,W^2\big),
\qquad
\tilde P:=X^{-1/2}PX^{-1/2}\succ 0,
\quad
W:=X^{-1/2}UX^{-1/2},
\]
which is positive whenever $U\neq 0$, since then $W^2\succeq 0$ is nonzero. So $X\mapsto\operatorname{Tr}(PX^{-1})$ is strictly convex on $\{X\succ 0\}$, and similarly for $Y\mapsto\operatorname{Tr}(QY^{-1})$; all other terms of $\Lambda(\cdot\,;P,Q)$ are affine in $(X,Y)$. Hence $\Lambda(\cdot\,;P,Q)$ is strictly convex in $(X,Y)$ and has at most one minimizer over $\{X,Y\succ 0\}$. By (a), the $(X,Y)$-part of every minimizer of \eqref{eq:matrix_potential} is this minimizer, and then $t$ is determined by \eqref{eq:kkt-tight}. This proves uniqueness of the minimizer.

For later use, we record the Hessian of $\Lambda(\cdot\,;P,Q)$ with respect to $(X,Y)$. Here and below, pairs of Hermitian matrices carry the inner product $\langle (U,V),(U',V')\rangle=\langle U,U'\rangle+\langle V,V'\rangle$, the induced Frobenius norm, and the component-wise semidefinite order. By \eqref{eq:inverse-expansion}, the Hessian is the block-diagonal map
\[
\mathcal H(U,V)
=\big(X^{-1}UX^{-1}PX^{-1}+X^{-1}PX^{-1}UX^{-1},\;\;
Y^{-1}VY^{-1}QY^{-1}+Y^{-1}QY^{-1}VY^{-1}\big),
\]
and the computation above shows that
\begin{equation}\label{eq:hessian-pd}
\begin{aligned}
\langle (U,V),\mathcal H(U,V)\rangle
&=2\operatorname{Tr}(PX^{-1}UX^{-1}UX^{-1})+2\operatorname{Tr}(QY^{-1}VY^{-1}VY^{-1})\\
&\ge\frac{2\varepsilon^2}{K}\big(\|U\|_{\mathrm F}^2+\|V\|_{\mathrm F}^2\big),
\end{aligned}
\end{equation}
with strict positivity whenever $(U,V)\neq 0$; for the quantitative bound in \eqref{eq:hessian-pd}, $P\succeq\varepsilon X^2$ gives
\[
\operatorname{Tr}(PX^{-1}UX^{-1}UX^{-1})
\ge\varepsilon\operatorname{Tr}(UX^{-1}U)
\ge\frac{\varepsilon^2}{K}\|U\|_{\mathrm F}^2 ,
\]
where the second inequality uses $X^{-1}\succeq(\varepsilon/K)I$ from \eqref{eq:box}; the same holds for $Y$, $Q$ and $V$.

\smallskip
\noindent\textbf{The Linearized Constraints.}
Now we prove the smoothness of the minimizer and $R$ via the implicit function theorem. 
This requires showing that the Jacobian of optimality conditions \eqref{eq:kkt-stationarity} and \eqref{eq:kkt-tight} is nonsingular. 
Define the linear maps on pairs of Hermitian matrices
\[
\mathcal L(U,V)
=\big(X^{-1}UX^{-1}-\eta(V),\;Y^{-1}VY^{-1}-\eta(U)\big),
\qquad
\mathcal T(U,V)
=\big(X\eta(V)X,\;Y\eta(U)Y\big).
\]
By \eqref{eq:inverse-expansion}, $-\mathcal L$ is the derivative of the left-hand sides of \eqref{eq:kkt-tight} with respect to $(X,Y)$. The map $\mathcal L$ is self-adjoint, by cyclicity of the trace and the self-adjointness of $\eta$. The map $\mathcal T$ is positive, hence monotone, because $\eta$ is. Multiplying the two components of the equation $\mathcal L(U,V)=(F,G)$ on both sides by $X$ and by $Y$, respectively, shows that
\begin{equation}\label{eq:L-vs-T}
\mathcal L(U,V)=(F,G)
\qquad\Longleftrightarrow\qquad
(\mathcal I-\mathcal T)(U,V)=(XFX,\,YGY),
\end{equation}
where $\mathcal I$ denotes the identity map. The key step in the proof of smoothness is the invertibility of $\mathcal L$, which we obtain by showing that $\mathcal T$ is a contraction in the semidefinite order.

\begin{claim}\label{claim:L-invertible}
$\mathcal L$ is invertible, with $\mathcal L^{-1}(F,G)=\sum_{j\ge 0}\mathcal T^j(XFX,YGY)$; in particular, $\mathcal L^{-1}$ maps pairs of PSD matrices to pairs of PSD matrices. Moreover, $\|\mathcal L^{-1}\|\le\varepsilon^{-1}$ and $\|\mathcal L\|\le 4K^2+c\nu$ in the operator norm induced by the Frobenius norm.
\end{claim}
\begin{proof}
By \eqref{eq:kkt-stationarity}, $\mathcal T(P,Q)=(P,Q)-\varepsilon(X^2,Y^2)$. By \eqref{eq:box} and $P\preceq I$, we have $\varepsilon X^2\succeq\frac{\varepsilon}{4K^2}I\succeq\frac{\varepsilon}{4K^2}P$, and likewise for $Y$ and $Q$. Hence
\[
\mathcal T(P,Q)\preceq r\,(P,Q),
\qquad
r:=1-\frac{\varepsilon}{4K^2}\in(0,1).
\]
Since $P,Q\succ 0$, every Hermitian pair $W$ satisfies $-a(P,Q)\preceq W\preceq a(P,Q)$ for some $a\ge 0$. Applying the monotone map $\mathcal T$ repeatedly gives $-ar^j(P,Q)\preceq\mathcal T^jW\preceq ar^j(P,Q)$ for all $j\ge 0$, so $\mathcal T^jW\to 0$ geometrically for every $W$. Therefore the series $\sum_{j\ge 0}\mathcal T^j$ converges, and the identity $(\mathcal I-\mathcal T)\sum_{j<m}\mathcal T^j=\mathcal I-\mathcal T^m\to\mathcal I$ shows that $\mathcal I-\mathcal T$ is invertible with inverse $\sum_{j\ge 0}\mathcal T^j$. By \eqref{eq:L-vs-T}, $\mathcal L$ is invertible with the stated inverse. Each term of the series maps PSD pairs to PSD pairs, so $\mathcal L^{-1}$ does too.

For the norm of $\mathcal L^{-1}$, note that \eqref{eq:kkt-stationarity} says exactly that $\mathcal L(P,Q)=(\varepsilon I,\varepsilon I)$, so $\mathcal L^{-1}(I,I)=\varepsilon^{-1}(P,Q)\preceq\varepsilon^{-1}(I,I)$. Since $\mathcal L^{-1}$ is positive, hence monotone, every pair $W$ with $-(I,I)\preceq W\preceq(I,I)$ satisfies $-\varepsilon^{-1}(I,I)\preceq\mathcal L^{-1}W\preceq\varepsilon^{-1}(I,I)$. Now $\mathcal L^{-1}$ is self-adjoint, so its operator norm is the largest absolute value of an eigenvalue $\theta$. Taking $W$ to be an eigenvector for $\theta$, normalized so that both components have operator norm at most $1$ and one of them has operator norm exactly $1$, the sandwich just derived for $\mathcal L^{-1}W=\theta W$ gives $|\theta|\le\varepsilon^{-1}$. The same argument, applied to the positive self-adjoint map $\eta$ and the bound $\eta(I)\preceq c\nu I$, gives $\|\eta\|\le c\nu$. Since $\|X^{-1}\|,\|Y^{-1}\|\le 2K$ by \eqref{eq:box}, the map $(U,V)\mapsto(X^{-1}UX^{-1},Y^{-1}VY^{-1})$ has norm at most $4K^2$, and thus $\|\mathcal L\|\le 4K^2+c\nu$.
\end{proof}

The norm bounds in Claim \ref{claim:L-invertible} are not needed for the present lemma, but they are useful for bounding the derivatives of $R$ quantitatively. The claim also shows that the multipliers are unique: if $(\hat P,\hat Q)$ is another pair of multipliers, then it satisfies \eqref{eq:kkt-stationarity} with the same $(X,Y)$, so $(P-\hat P,\,Q-\hat Q)$ is a fixed point of $\mathcal T$, and $\mathcal T^j\to 0$ forces $\hat P=P$ and $\hat Q=Q$.

\smallskip
\noindent\textbf{Smoothness.}
Write $z=(t,X,Y,P,Q)$, and let $F(x;z)$ be the map whose components are the differences of the two sides of the five equations in \eqref{eq:kkt-stationarity} and \eqref{eq:kkt-tight}, with the two matrix equations in \eqref{eq:kkt-stationarity} written as $X^{-1}PX^{-1}-\eta_x(Q)-\varepsilon I=0$ and $Y^{-1}QY^{-1}-\eta_x(P)-\varepsilon I=0$. Then $F$ takes values in $\mathbb R\times\operatorname{Herm}_d^4$, where $\operatorname{Herm}_d$ is the real vector space of $d\times d$ Hermitian matrices, of dimension $d^2$; this is a space of the same dimension $1+4d^2$ as the space of unknowns $z$. We have shown that for every $x$, the unique minimizer together with its unique multipliers is a solution of $F(x;z)=0$ with $X,Y,P,Q\succ 0$. Now restrict $x$ to an open face of the cube. On the face, the frozen coordinates are constant and the active coordinates lie in $(-1,1)$, where $\psi_i$ is smooth, so $S(x)$ and $\eta_x$ depend smoothly on $x$, and $F$ is smooth in $(x,z)$ on the open set where $X,Y\succ 0$.

Fix $x_0$ in the face and let $z_0=(t,X,Y,P,Q)$ be the corresponding solution. We claim that the Jacobian $D_zF(x_0;z_0)$ is nonsingular. Let $(h,U,V,P',Q')$ lie in its kernel. Differentiating the five equations in the direction $(h,U,V,P',Q')$, using \eqref{eq:inverse-expansion}, gives
\[
\mathcal L(U,V)=-h\,(I,I),
\qquad
\mathcal L(P',Q')=\mathcal H(U,V),
\qquad
\operatorname{Tr}(P'+Q')=0,
\]
where the first system is the derivative of \eqref{eq:kkt-tight} and the second is the derivative of the two matrix equations in \eqref{eq:kkt-stationarity}. Pairing the middle system with $(U,V)$ and using the self-adjointness of $\mathcal L$,
\[
\langle (U,V),\mathcal H(U,V)\rangle
=\langle \mathcal L(U,V),(P',Q')\rangle
=-h\operatorname{Tr}(P'+Q')
=0 .
\]
By \eqref{eq:hessian-pd}, this forces $U=V=0$. The first system then gives $h=0$, and the middle system becomes $\mathcal L(P',Q')=0$, so $P'=Q'=0$ by Claim \ref{claim:L-invertible}. The kernel is trivial, as claimed.

By the implicit function theorem, there are a neighborhood $\mathcal U$ of $x_0$ in the face and a smooth map $x\mapsto z(x)=(t(x),X(x),Y(x),P(x),Q(x))$ on $\mathcal U$ with $z(x_0)=z_0$ and $F(x;z(x))=0$. Shrinking $\mathcal U$, we may assume by continuity that $X(x),Y(x),P(x),Q(x)\succ 0$ on $\mathcal U$. We claim that $(t(x),X(x),Y(x))$ is the minimizer of \eqref{eq:matrix_potential} at $x$ for every $x\in\mathcal U$. It is feasible, since it satisfies \eqref{eq:kkt-tight} at $x$. Let $(t',X',Y')$ be any feasible point at $x$, and write $\Lambda(\cdot)$ for $\Lambda(\cdot\,;P(x),Q(x))$ at $x$. Then
\[
t'+\varepsilon\operatorname{Tr}(X'+Y')
\;\ge\;\Lambda(t',X',Y')
\;\ge\;\Lambda(t(x),X(x),Y(x))
\;=\;t(x)+\varepsilon\operatorname{Tr}(X(x)+Y(x)) .
\]
The first inequality holds because the constraint matrices at $(t',X',Y')$ are negative semidefinite and $P(x),Q(x)\succeq 0$. The second holds because $\Lambda$ is convex, as $P(x),Q(x)\succeq 0$, and its gradient vanishes at $(t(x),X(x),Y(x))$ by \eqref{eq:kkt-stationarity} at $x$, so this point is a global minimizer of $\Lambda$. The final equality holds because the constraints are tight at $(t(x),X(x),Y(x))$. Thus $(t(x),X(x),Y(x))$ is optimal, hence equal to the unique minimizer at $x$, and $(P(x),Q(x))$ are the unique multipliers at $x$. Consequently the minimizer, the multipliers, and $R(x)=t(x)+\varepsilon\operatorname{Tr}(X(x)+Y(x))$ are smooth on $\mathcal U$. Since $x_0$ was arbitrary, they are smooth on the whole open face.
\end{proof}
\subsection{Proof of \texorpdfstring{\Cref{lem:potential_decrease}}{Lemma 3.2}} \label{sec3.2}
Recall the large constant $c = 40$ in the definition of $\eta_x(Z) = c\sum_i \psi_i A_i Z A_i$. 
All unqualified implicit constants are absolute and independent of $c$. We first restate the lemma for completeness.

\PotentialDecrease*

Fix a state $x$, relabel its active coordinates as $[n]$, and suppress the dependence on $x$; all derivatives below are taken within the current open face of $[-1,1]^n$ on which $x$ lies. Let $(t, X, Y)$ minimize the matrix potential $R$, and write
\[
A_i = v_i v_i^*, \qquad a_i =v_i^*Xv_i, \qquad b_i = v_i^*Y v_i.
\]
Suppose $c\psi_ib_i \geq 1-x_i$ for some active coordinate $i$. Moving $x_i$ to $+1$ adds $(1-x_i)A_i$ to $S$ and removes the $i$-th summand $c\psi_ib_i$ from $\eta_x$, and so the two constraint matrices in \eqref{eq:matrix_potential} change by $(1-x_i-c\psi_ib_i)A_i$ and $-(1-x_i+c\psi_ia_i)A_i$ respectively. Note further that both changes are negative semidefinite, so the same $(t, X, Y)$ remains feasible after the entire endpoint move. Thus $R$ does not increase, while $\Phi$ decreases by $\psi_i$, and hence the potential $\Psi$ decreases. Similarly, $c\psi_ia_i \geq 1+x_i$ permits the move to $-1$. We may therefore assume that for all $i \in [n]$,
\begin{equation}\label{eq6}
    c\psi_i a_i < 1 + x_i, \qquad c\psi_ib_i < 1-x_i.
\end{equation}
Under these conditions, we want to find $\mu$ and $\Sigma$ that decrease $\Psi$ as in \Cref{lem:potential_decrease}.

\smallskip
\noindent \textbf{Overview.} We first explain the high-level intuition for why, under \eqref{eq6}, a drift $\mu$ and a covariance $\Sigma$ that decrease $\Psi$ should exist. First consider moving $x$ in a random direction $h$ with $\mathbb{E}[h] = 0$; to second order, the matrix potential $R$ changes for two competing reasons  (see \Cref{lem:derivatives-R}):
\begin{itemize}
    \item The barrier terms $X^{-1}, Y^{-1}$ in \eqref{eq:matrix_potential} are convex, so the random motion increases $R$. This cost is a nonnegative quadratic form (denoted by $E_h$) in the derivatives $(\dot{X}, \dot{Y})$ of $(X, Y)$ along $h$.
    \item The reservoir terms $\eta_x(Y)$ and $\eta_x(X)$ act in the opposite direction. Since $\psi$ in $\eta_x$ is concave, a coordinate that moves randomly loses reservoir in expectation as
    \[
    \mathbb E\psi(x_i + h_i) \approx \psi_i + \frac{1}{2}\psi_i''\mathbb{E}[h_i^2] < \psi_i.
    \]
    A smaller reservoir loosens both constraints in \eqref{eq:matrix_potential}, so $R$ decreases. Furthermore, the reservoir carries the factor $c$ from $\eta_x$ and the barrier does not, so once both are measured with the same weights, the gain exceeds the cost by a factor of order $c$.
\end{itemize}
Note that $E_h$ does not split over coordinates. The reservoirs couple the coordinates so that moving $x_i$ changes $X$ which changes $\eta_x(X)$ and hence $Y$, which changes $\eta_x(Y)$ and hence $X$, and so on; consequently, a badly chosen $h$ could make this feedback quite large. We therefore take $h$ random in a subspace, still of dimension about $n$, on which the feedback is a contraction, with a covariance dominated by its diagonal (via Lemma \ref{lem:covariance}). For such $h$, using the conditions \eqref{eq6}, the expected energy cost $\mathbb{E}[E_h]$ is at most a constant times a sum of per-coordinate terms (Lemma \ref{lem:covariance-bound}).

Furthermore, even for a single coordinate, the cost and the gain are weighted differently (see \eqref{eq21}), since 
the reservoir in each constraint is built from the other constraint's variable. So the gain cannot pay for the cost coordinate by coordinate. The mismatch, however, 
is a multiple of $\partial_i R$ (see \eqref{eq:straight-crossed}), and hence a drift down this slope makes up the difference and decreases $R$ (see \eqref{eq22}).

With both ingredients above, $R$ does not increase in expectation, yet there's a definite decrease coming from the fact that random motion lowers $\Phi(x) = \sum_i \psi_i(x)$ at rate $\frac{1}{2}(-\psi_i'')\mathbb{E}[h_i^2]$; under \eqref{eq6}, the drift is small enough to use up only an $O(1/c)$ fraction of this decrease, resulting in an overall potential decrease. Now we give the details.  

\smallskip
\noindent \textbf{A Feasible Path for $X$ and $Y$.} Let $P, Q \succ 0$ be the optimal dual matrices, and set $p_i = v_i^*Pv_i$ and $q_i = v_i^*Qv_i$. As in our proof of \Cref{lem:matrix_potential} we use the stationarity identities
\begin{equation}\label{eq7}
\operatorname{Tr}(P+Q) = 1, \qquad
P = X\eta_x(Q)X + \eps X^2, \qquad
Q = Y\eta_x(P)Y + \eps Y^2,
\end{equation}
together with the fact that both of the primal constraints are tight as in \eqref{eq:kkt-tight}. Now recall the constraint linearization from the proof of \Cref{lem:matrix_potential},
\[
\mathcal{L}(U, V) = (X^{-1}UX^{-1} - \eta_x(V), Y^{-1}VY^{-1} - \eta_x(U)),
\]
which is the (negation of) derivatives of the constraints \eqref{eq:kkt-tight} with respect to $(X,Y)$ and is invertible by Claim \ref{claim:L-invertible}. For any Hermitian $U, V$, the self-adjointness of $\eta_x$ and \eqref{eq7} give
\begin{align}\label{eq8}
    \langle (P, Q), \mathcal{L}(U, V) \rangle &= \operatorname{Tr}((X^{-1}PX^{-1} - \eta_x(Q))U) + \operatorname{Tr}((Y^{-1}QY^{-1} - \eta_x(P))V) \notag \\
    &= \eps \operatorname{Tr}(U + V).
\end{align}
Thus, pairing the derivatives $\mathcal{L}(U, V)$ with $(P, Q)$ gives the corresponding derivative of the objective term $\eps \operatorname{Tr}(X + Y)$. We use this for the first and second derivatives along a feasible path.

Fix a direction $h \in \R^n$. Keep $t$ fixed and solve the tight constraints near the current state:
\begin{align} 
    X(s)^{-1} + S(x + sh) + \eta_{x+sh}(Y(s)) &= tI, \label{eq:feasible-path-1}\\
    Y(s)^{-1} - S(x + sh) + \eta_{x+sh}(X(s)) &= tI, \label{eq:feasible-path-2}
\end{align}
with $X(0)=X$ and $Y(0)=Y$. The derivative of this system with respect to its matrix variables at $s = 0$ is $-\mathcal{L}$. The implicit-function theorem therefore gives a smooth local solution, and $X(s), Y(s)$ remain positive definite for sufficiently small $|s|$.
\begin{remark}
    This path is feasible but need not be optimal at the perturbed states as $t$ is fixed; this is sufficient because its objective value will give an upper bound on $R(x +sh)$.
\end{remark}
Throughout the calculations below, keep $P, Q$ fixed at their values at $x$, and write $\Dot{X} = X'(0)$ and $\Dot{Y} = Y'(0)$. Differentiating the first constraint \eqref{eq:feasible-path-1} gives
\[
-X^{-1}\Dot{X}X^{-1} + \sum_{i \in [n]} h_i A_i + c \sum_{i \in [n]} \psi_i' h_i A_i Y A_i + \eta_x(\Dot{Y}) = 0.
\]
Since $A_iYA_i = b_i A_i$ and similarly $A_i X A_i = a_i A_i$, the two differentiated constraints become 
\begin{equation}\label{eq:path-solution}
\begin{aligned}
X^{-1}\dot{X}X^{-1}
&= \sum_i \alpha_i h_i A_i + \eta_x(\dot{Y}), \\
Y^{-1}\dot{Y}Y^{-1}
&= -\sum_i \beta_i h_i A_i + \eta_x(\dot{X}),
\end{aligned}
\qquad
\alpha_i = 1 + c\psi_i' b_i,
\quad
\beta_i = 1 - c\psi_i' a_i.
\end{equation}
Note that $\alpha_i$ and $\beta_i$ are the net rates at which $x_i$ pushes the two constraints: increasing $x_i$ adds $h_iA_i$ to $S$, and also changes the reservoir term $c\psi_ib_iA_i$ of the first constraint at rate $c\psi_i'b_ih_i$ (and the term $c\psi_ia_iA_i$ of the second at rate $c\psi_i'a_ih_i$). 
The complementary inequalities \eqref{eq6} ensure that these coefficients are positive. Indeed, if $x_i \geq 0$ then $\psi_i' \leq 0$ and \eqref{eq6} gives the bounds 
\begin{equation}\label{eq11}
    \frac{1 + x_i/3}{1 + x_i} \leq \alpha_i \leq 1,
    \qquad
    1 \leq \beta_i \leq \frac{1 - x_i/3}{1 - x_i}.
\end{equation}
If $x_i < 0$, then $\psi_i' > 0$ and \eqref{eq6} gives the alternative bounds. 
\begin{equation*}
  1 \leq \alpha_i \leq \frac{1 + x_i/3}{1 + x_i}, \qquad \frac{1 - x_i/3}{1 - x_i} \leq \beta_i \leq 1 .
\end{equation*}
In particular, $\alpha_i\beta_i = \Omega(1)$ for every active coordinate $i$.

It follows from \eqref{eq:path-solution} that $(\Dot{X}, \Dot{Y})$ depends linearly on $h$. Define the \emph{barrier energy} as 
\[
E_h := \operatorname{Tr}(PX^{-1}\Dot{X}X^{-1}\Dot{X}X^{-1})
+ \operatorname{Tr}(QY^{-1}\Dot{Y}Y^{-1}\Dot{Y}Y^{-1}) ,
\]
which is clearly a nonnegative quadratic form in $h$; it is the second-order cost of the barrier terms, which the reservoir gain must pay for.
Here, $X^{-1}\dot{X}X^{-1}\dot{X}X^{-1}$ comes from the second derivative of $X^{-1}$, and $E_h$ pairs this (and its analogue for $Y$) with their corresponding dual matrices. 
\begin{lemma}[First and second variations]\label{lem:derivatives-R}
    For every active coordinate $i$ and every direction $h \in \R^n$,
    \begin{equation}\label{eq10}
    \partial_i R = \alpha_i p_i - \beta_i q_i,
    \qquad
    h^\top (\nabla^2 R) h \leq \frac{11}{4}E_h - c\sum_{i \in [n]} \chi_ih_i^2(b_ip_i + a_iq_i),
    \end{equation}
    where $\chi_i=-\psi_i''-\frac{4}{3}(\psi_i')^2/\psi_i > 0$. 
\end{lemma}
\begin{proof}
    Let $J(s) = t + \eps \operatorname{Tr}(X(s) + Y(s))$. By the feasible-path construction in \eqref{eq:feasible-path-1} and \eqref{eq:feasible-path-2}, we have $J(s) \geq R(x + sh)$, with equality at $s = 0$. Hence the smooth function $J(s) - R(x + sh)$ has a local minimum at zero, and so
    \[
    J'(0) = \nabla R \cdot h, \qquad
    J''(0) \geq h^\top \nabla^2 R h.
    \]
    Applying \eqref{eq8} to $(\Dot{X}, \Dot{Y})$, and then using \eqref{eq:path-solution}, gives
    \[
    J'(0) = \eps\operatorname{Tr}(\Dot{X} + \Dot{Y})
    = \langle (P, Q), \mathcal{L}(\Dot{X}, \Dot{Y}) \rangle
    = \sum_{i \in [n]} (\alpha_i p_i - \beta_i q_i) h_i.
    \]
    Since this holds for every $h$, the gradient formula follows.

    For the second derivative, write $\Ddot{X} = X''(0)$ and $\Ddot{Y} = Y''(0)$. Differentiating the first tight constraint \eqref{eq:feasible-path-1} twice gives
    \[
    X^{-1}\Ddot{X}X^{-1} - \eta_x(\Ddot{Y}) =
    2X^{-1}\Dot{X}X^{-1}\Dot{X}X^{-1} + 2c\sum_{i\in [n]} \psi_i' h_i A_i \Dot{Y} A_i + c \sum_{i \in [n]} \psi_i''h_i^2 b_i A_i.
    \]
    The second constraint gives the same equation with $X, Y$ interchanged and $b_i$ replaced by $a_i$.
    
    Now pair these two equations with $P, Q$. By \eqref{eq8}, their left-hand sides sum to $\eps\operatorname{Tr}(\Ddot{X} + \Ddot{Y}) = J''(0)$. On the right, the inversion terms sum to $2E_h$, while the rank-one identities give
    \[
    J''(0) = 2E_h + 2c\sum_{i \in [n]} \psi_i'h_i(p_i v_i^* \Dot{Y} v_i + q_i v_i^* \Dot{X} v_i) + c\sum_{i\in [n]} \psi_i'' h_i^2 (b_i p_i + a_i q_i).
    \]
    The last term is negative; we must bound the mixed term without using all of this negative contribution. For the $X$-part of $E_h$, \eqref{eq7} gives
    \[
    \operatorname{Tr}(PX^{-1}\Dot{X}X^{-1}\Dot{X}X^{-1})
    = \operatorname{Tr}((\eta_x(Q) + \eps I)\Dot{X}X^{-1}\Dot{X})
    \geq c\sum_{i \in [n]} \psi_i q_i v_i^*\Dot{X}X^{-1}\Dot{X}v_i,
    \]
    where we dropped the nonnegative term $\eps\operatorname{Tr}(\Dot{X}X^{-1}\Dot{X})$. We may likewise apply the same calculation to the $Y$-part to get
    \[
    E_h \geq c\sum_{i \in [n]} \psi_i (p_i v_i^*\Dot{Y}Y^{-1}\Dot{Y}v_i + q_i v_i^* \Dot{X}X^{-1}\Dot{X}v_i).
    \]
    By Cauchy-Schwarz and writing $v_i^*\dot{X}v_i = \langle X^{-1/2}\dot{X}v_i, X^{1/2}v_i\rangle$, we get
    \[
    |v_i^* \Dot{X} v_i|^2 \leq a_i v_i^* \Dot{X}X^{-1}\Dot{X}v_i,
    \qquad
    |v_i^* \Dot{Y} v_i|^2 \leq b_i v_i^* \Dot{Y}Y^{-1}\Dot{Y}v_i .
    \]
    Write $e_i^X = v_i^*\dot{X}X^{-1}\dot{X}v_i$ and $e_i^Y = v_i^*\dot{Y}Y^{-1}\dot{Y}v_i$, so that $E_h \geq c \sum_i \psi_i(p_ie_i^Y + q_ie_i^X)$. By the last two inequalities, the absolute value of the mixed term is at most
    \begin{align*}
    2c\sum_{i\in[n]}|\psi_i'h_i|\Bigl(p_i\sqrt{b_ie^Y_i} + q_i\sqrt{a_ie^X_i}\Bigr)
    &= 2c\sum_{i\in[n]}\Bigl(\sqrt{\psi_ip_ie^Y_i}\cdot\frac{|\psi_i'h_i|}{\sqrt{\psi_i}}\sqrt{b_ip_i} + \sqrt{\psi_iq_ie^X_i}\cdot\frac{|\psi_i'h_i|}{\sqrt{\psi_i}}\sqrt{a_iq_i}\Bigr)\\
    &\le 2\Bigl(c\sum_{i\in[n]}\psi_i(p_ie^Y_i + q_ie^X_i)\Bigr)^{1/2}\Bigl(c\sum_{i\in[n]}\frac{(\psi_i')^2}{\psi_i}h_i^2(b_ip_i + a_iq_i)\Bigr)^{1/2}\\
    &\le 2\sqrt{E_hB_h},
    \end{align*}
    where the second step is Cauchy-Schwarz over the $2n$ products and the last step uses the lower bound on $E_h$. Here,
    \[
    B_h := c\sum_{i \in [n]} \frac{(\psi_i')^2}{\psi_i}h_i^2(b_ip_i + a_i q_i).
    \]
    For $\psi(x) = (1-x^2)^{1/3}$, we have
    \begin{equation}\label{eq:post-hoc}
    -\psi''(x) = \frac{2}{3}(1-x^2)^{-5/3}\biggl(1 + \frac{x^2}{3}\biggr),
    \qquad
    \frac{\psi'(x)^2}{\psi(x)(-\psi''(x))}
    = \frac{2x^2}{3 + x^2} \leq \frac{1}{2}.
    \end{equation}
    In particular $\frac{4}{3}(\psi')^2/\psi \leq \frac{2}{3}(-\psi'')$, so $\chi(x)=\frac{2}{3}(1-x^2)^{-5/3}(1-\frac{5}{9}x^2) > 0$.
    Using $2\sqrt{E_hB_h} \leq \frac{3}{4}E_h + \frac{4}{3}B_h$ in the expression for $J''(0)$, we obtain
    \[
    J''(0) \leq \frac{11}{4}E_h - c\sum_{i \in [n]} \biggl(-\psi_i'' - \frac{4}{3}\frac{(\psi_i')^2}{\psi_i}\biggr)h_i^2(b_ip_i + a_iq_i) = \frac{11}{4}E_h - c\sum_{i \in [n]} \chi_ih_i^2(b_ip_i + a_iq_i)
    \]
    Together with $h^\top (\nabla^2 R) h \leq J''(0)$, this completes the proof of \eqref{eq10}.
\end{proof}
\begin{remark}
    Note that the exponent $1/3$ in \eqref{eq:coordinate_potential} enters through the ratio in \eqref{eq:post-hoc}. For $\psi = (1-x^2)^\rho$, this ratio equals $2\rho x^2/(1+(1-2\rho)x^2)$, which stays bounded away from $1$ on $[-1, 1]$ when $\rho < 1/2$, which is all the argument needs. Taking $\rho = 1/3$ is merely a convenient choice.
\end{remark}
By Lemma \ref{lem:derivatives-R}, it suffices to find a random $h$ for which $\mathbb E[E_h]$ is small compared with the reservoir gain in the second term of \eqref{eq10}. The difficulty here, however, is that $E_h$ is a quadratic form in the solution $(\dot{X}, \dot{Y})$ of \eqref{eq:path-solution}, which is a linear system for two $d \times d$ matrices. We now show that it is really a linear system for $2n$ variables.

The key observation is that $\eta_x$ maps every matrix into the span of the $A_i$: since $A_iZA_i=(v_i^*Zv_i)A_i$,
\begin{equation}\label{eq:eta-x}
    \eta_x(Z) = c\sum_{i \in [n]}\psi_i(v_i^*Zv_i)A_i.
\end{equation}
So the right-hand sides of both equations in \eqref{eq:path-solution} lie in the span of the $A_i$, and hence so do $X^{-1}\dot{X}X^{-1}$ and $Y^{-1}\dot{Y}Y^{-1}$. We may therefore write
\begin{equation}\label{eq:expanded-sum}
    X^{-1}\dot{X}X^{-1} = \sum_{i \in [n]}\ell_iF_iA_i, \qquad
    Y^{-1}\dot{Y}Y^{-1} = \sum_{i \in [n]}m_iG_iA_i
\end{equation}
for coefficient vectors $F, G \in \R^n$ that we will choose.\footnote{Here, $\ell_i, m_i$ are scalars that put the two constraints on a common footing by splitting the reservoir weight $\psi_i = \ell_i m_i$ between the two sides in the ratio $\alpha_i : \beta_i$ at which $x_i$ pushes them. Even though they can, in principle, be absorbed into the coefficient vectors $F$ and $G$, we write them out explicitly so that the linear system can be put into a more symmetric form as in \eqref{eq:reformed} below.} By \eqref{eq:expanded-sum}, $\dot{Y}=\sum_jm_jG_jYA_jY$, so $v_i^*\dot{Y}v_i=\sum_jm_j|v_i^*Yv_j|^2G_j$. It then follows from \eqref{eq:eta-x} that
\[
\eta_x(\dot{Y})=c\sum_{i\in [n]}\psi_i\biggl(\sum_{j \in [n]}m_j|v_i^*Yv_j|^2G_j\biggr)A_i,\qquad
\eta_x(\dot{X})=c\sum_{i\in [n]}\psi_i\biggl(\sum_{j \in [n]}\ell_j|v_i^*Xv_j|^2F_j\biggr)A_i.
\]
Substituting these and \eqref{eq:expanded-sum} into \eqref{eq:path-solution}, and matching the coefficient of each $A_i$ (and dividing by $\ell_i$, resp. $m_i$), we obtain the system
\begin{equation}\label{eq:turned}
    F_i - \frac{c\psi_i}{\ell_i}\sum_j m_j |v_i^* Y v_j|^2 G_j = \frac{\alpha_i}{\ell_i}h_i,
    \qquad
    \frac{c\psi_i}{m_i}\sum_j\ell_j|v_i^*Xv_j|^2F_j - G_i = \frac{\beta_i}{m_i}h_i.
\end{equation}
In particular, the response $F_i$ of the first barrier in direction $v_i$ is the direct push $\alpha_ih_i$ plus feedback from the second barrier through the reservoir, where the factor $|v_i^*Yv_j|^2$ measures how much a response in direction $v_j$ changes the reservoir term of coordinate $i$. (The second equation is the same with the roles of the two blocks exchanged.)
\begin{remark}
    The $A_i$ may be linearly dependent, so \eqref{eq:turned}
    is sufficient but not necessary for \eqref{eq:path-solution}. But this is not an issue; whenever $(F, G)$ solves \eqref{eq:turned}, the matrices defined by \eqref{eq:expanded-sum} solve \eqref{eq:path-solution}, and since $\mathcal{L}$ is invertible (Claim \ref{claim:L-invertible}) they are precisely $(X'(0),Y'(0))$ for this $h$.
\end{remark}

To further simplify this linear system, we choose the positive scalars $T_i, \ell_i$ and $m_i$ satisfying
\begin{equation}\label{eq:scales}
    \alpha_i = T_i \ell_i, \qquad \beta_i = T_i m_i, \qquad \ell_i m_i = \psi_i .
\end{equation}
Solving these equations give  
\[
\ell_i = \sqrt{\psi_i\alpha_i / \beta_i}, \qquad
m_i = \sqrt{\psi_i\beta_i / \alpha_i}, \qquad
T_i = \sqrt{\alpha_i\beta_i / \psi_i} .
\]
With these scales, $c\psi_im_j/\ell_i = cm_im_j$, $c\psi_i\ell_j/m_i = c\ell_i\ell_j$, and $\alpha_i/\ell_i = \beta_i/m_i = T_i$, so \eqref{eq:turned} becomes
\begin{equation}\label{eq:reformed}
    F-BG=H, \qquad AF-G = H, \qquad \text{where }H_i = T_ih_i,
\end{equation}
and $A, B$ are the $n \times n$ \emph{interaction matrices}
\[
A_{ij} = c\ell_i\ell_j|v_i^*Xv_j|^2, \qquad
B_{ij} = cm_im_j|v_i^*Yv_j|^2
\]
which carry the feedback between the two barriers through the reservoirs.
Both of these matrices are symmetric, entrywise nonnegative, and positive semidefinite as they can be expressed as Gram matrices.\footnote{The two conditions in \eqref{eq:scales} are precisely what we need for this form---with $\ell_im_i=\psi_i$, the matrices $A$ and $B$ are symmetric, and $\alpha_i/\ell_i = \beta_i/m_i$ ensures that both equations are driven by the same vector $H$.}
Intuitively, $H_i := T_i h_i$ will be the unit in which we measure movement in coordinate $i$; the choice of the scales $T_i, m_i$ and $\ell_i$ above makes the linear system \eqref{eq:reformed} more symmetric. 

We write $\widetilde{a}_i := \ell_ia_i$, $\widetilde{b_i} := m_ib_i$, $\widetilde{p}_i := \ell_i p_i$, $\widetilde{q}_i := m_i q_i$. The next lemma bounds these quantities.

\begin{lemma}[Scalar bounds]\label{lem4.2}
    Under conditions \eqref{eq6}, and with $\chi_i := -\psi_i''-\frac{4}{3}(\psi_i')^2/\psi_i$ defined as in Lemma \ref{lem:derivatives-R}, the rescaled quantities satisfy
    \begin{enumerate}
        \item[1.] $\widetilde{a}_i + \widetilde{b}_i = O(1/c)$ and $\widetilde{a}_i \widetilde{b}_i < 1/c^2$,
        \item[2.] $-\psi_i''/T_i^2 = \Omega(1)$, $1/T_i^2 = O(1)$, and $|\psi_i'|/\psi_i \leq \frac{1}{2}(-\psi_i'')$,
        \item[3.] $\chi_i/(\alpha_i\beta_i) \geq 5/9$.
    \end{enumerate}
\end{lemma}
\begin{proof}
    Fix a coordinate and omit its index for ease of notation. Replacing $x$ by $-x$ and interchanging $a, b$ exchanges $\alpha, \beta$ and preserves \eqref{eq6}, so we may assume $x \geq 0$ and put
    \[
    r := 1 -x^2, \qquad u := \frac{c\psi a}{1 + x}, \qquad v := \frac{c\psi b}{1-x}.
    \]
    By \eqref{eq6}, $u, v \in (0, 1)$. Since $\psi = r^{1/3}$ and $\psi' = -2x/(3r^{2/3})$, we have by the definitions of $u,v$ that $ca=(1+x)u/r^{1/3}$ and $cb=(1-x)v/r^{1/3}$. Substituting these into $\alpha=1+c\psi'b$ and $\beta=1-c\psi'a$ and using $r = (1-x)(1+x)$, we get
    \begin{align*}
        \alpha &= 1-\frac{2x(1-x)v}{3r} = \frac{1+x-2xv/3}{1+x},\\
        \beta &= 1+\frac{2x(1+x)u}{3r} = \frac{1-x+2xu/3}{1-x}.
    \end{align*}
    Consequently the definitions of $\widetilde{a},\widetilde{b}$ give
    \[
    c^2\widetilde{a}^2 = \psi\frac{\alpha}{\beta}(ca)^2
    = r^{1/3}\frac{(1-x)(1+x-2xv/3)}{(1+x)(1-x+2xu/3)}\frac{(1+x)^2u^2}{r^{2/3}}
    = r^{2/3}u^2\frac{1+x-2xv/3}{1-x+2xu/3},
    \]
    and similarly,
    \[
    c^2\widetilde{b}^2 = \psi\frac{\beta}{\alpha}(cb)^2
    = r^{1/3}\frac{(1+x)(1-x+2xu/3)}{(1-x)(1+x-2xv/3)}\frac{(1-x)^2v^2}{r^{2/3}}
    = r^{2/3}v^2\frac{1-x+2xv/3}{1+x-2xu/3}.
    \]
    The first expression increases with $u$ and decreases with $v$, and the second increases with both. Maximizing over $0 \leq u, v \leq 1$, we obtain
    \[
    c^2\widetilde{a}^2 \leq r^{2/3}\frac{1+x}{1-x/3} = O(1),
    \qquad
    c^2\widetilde{b}^2 \leq r^{2/3}\frac{1-x/3}{1+x/3} = O(1),
    \]
    and thus $\widetilde{a} + \widetilde{b} = O(1/c)$. Multiplying the two expressions gives $c^2\widetilde{a} \widetilde{b} = r^{2/3}uv < 1$.

    The bounds in \eqref{eq11} imply 
    \[
    \Omega(1) \leq \alpha \beta \leq \frac{1-x/3}{1-x} = \frac{(1-x/3) (1 + x)}{r} \leq O(1/r). 
    \]
    Since $\psi = r^{1/3} \leq 1$, this gives $T^{-2} = \psi/(\alpha \beta) = O(1)$. The first formula in \eqref{eq:post-hoc} gives  $-\psi'' = \Theta(r^{-5/3})$, and so $-\psi''/T^2= \psi(-\psi'')/(\alpha\beta)=\Omega(1/r^{1/3})=\Omega(1)$. By \eqref{eq:post-hoc}, 
    \[
    |\psi'|/(\psi(-\psi''))=\frac{2 x^2}{(3 + x^2) / |\psi'|}  = \frac{xr^{2/3}}{1+x^2/3}\leq x\sqrt{r} \leq 1/2.
    \]

    It remains to lower bound $\chi/(\alpha\beta)$. Recall that $0 \leq x < 1$, $\psi=(1-x^2)^{1/3}$, and 
    \[
    0 < \alpha\beta \leq \frac{1-x/3}{1-x}, \qquad \chi = \frac{2}{3}(1-x^2)^{-5/3}\left(1-\frac{5}{9}x^2\right) > 0.
    \]
    Consequently,
    \begin{align*}
        \frac{\chi}{\alpha\beta} &\geq \frac{2}{3}(1-x^2)^{-5/3}\left(1-\frac{5}{9}x^2\right)\frac{1-x}{1-x/3}= \frac{2(9-5x^2)}{9(1+x)(3-x)\psi^2}= \frac{10\psi + 8/\psi^2}{9(1+x)(3-x)},
    \end{align*}
    where the second equality uses $\psi^3=(1-x)(1+x)$, and the last uses $9-5x^2=4+5\psi^3$. Since the denominator is positive, this lower bound is at lest $5/9$ precisely when
    \[
    10\psi + \frac{8}{\psi^2} \geq 5(1+x)(3-x) = 20 - 5(1-x)^2.
    \]
    Thus it suffices to prove
    $
    5(1-x)^2 + 10\psi + {8}/{\psi^2} \geq 20.
    $
    For $0 \leq y < 1$, we first establish
    \[(1-y)^{-2/3} \geq \frac{6-y}{6-5y}.\]
    Both sides of this are positive, so taking cubes and clearing out the positive denominators reduces the inequality to
    \begin{align*}
        (6-5y)^3 - (1-y)^2(6-y)^3 &= y^3(y^2-20y+20)= y^3(y^2 + 20(1-y)) \geq 0.
    \end{align*}
    Now take $y=x^2$. Since
    \[
    10\psi + \frac{8}{\psi^2} = (18-10x^2)(1-x^2)^{-2/3},
    \]
    and $18-10x^2>0$, the preceding bound gives
    \begin{align*}
        5(1-x)^2 + 10\psi + \frac{8}{\psi^2} - 20 &\geq 5(1-x^2) - 20 + \frac{(18-10x^2)(6-x^2)}{6-5x^2}\\
        &= \frac{18-60x+27x^2+50x^3-15x^4}{6-5x^2}
    \end{align*}
    Clearly the denominator is positive. To bound the numerator, write $p(x) = 18-60x+27x^2+50x^3-15x^4$ and set $t = 2x-1\in [-1,1)$. Expanding this gives
    \begin{align*}
        16p(x)&= 1-24t+318t^2+40t^3-15t^4\\
        &= (1-12t)^2 + t^2(174+40t-15t^2)\geq (1-12t)^2 + 119t^2 > 0,
    \end{align*}
    where we've used $|t| \leq 1$. Therefore $p(x)>0$, and hence $\chi/(\alpha\beta) \geq 5/9$.
\end{proof}

\smallskip
\noindent \textbf{The Choice of \texorpdfstring{$(F, G)$}{(F,G)}.}
Before proceeding, we first briefly explain why we choose $(F, G)$ rather than the more immediate and obvious choice $h$. Observe that the map $h \mapsto (F, G)$ is linear, and its image is the set of pairs satisfying the consistency condition $F-BG=AF-G$; given such a pair, \eqref{eq:reformed} defines $H$ and hence $h$ (and as we see below in Lemma \ref{lem:covariance-bound}, \eqref{eq:reformed} has a unique solution $(F, G)$ for each $H$). So choosing $(F, G)$ is the same as choosing $h$, and doing so further accrues the advantage that $E_h$ is simple in terms of $(F, G)$: by \eqref{eq:expanded-sum}, $X^{-1}\dot{X}X^{-1}\dot{X}X^{-1} = C_FXC_F$ with $C_F := \sum_i\ell_iF_iA_i$, and similarly for $Y$ with $C_G := \sum_i m_iG_iA_i$, so
\begin{equation}\label{eq:again-rewritten}
    E_h = \operatorname{Tr}(PC_FXC_F) + \operatorname{Tr}(QC_GYC_G).
\end{equation}
This is a positive semidefinite quadratic form in $(F, G)$ whose coefficient on the diagonal $F_i^2$ is $\ell_i^2\operatorname{Tr}(PA_iXA_i)=\ell_i^2a_ip_i=\widetilde{a}_i\widetilde{p}_i$, and whose coefficient on $G_i^2$ is $\widetilde{b}_i\widetilde{q}_i$. As we will see below in \Cref{lem:covariance-bound}, these diagonal coefficients are exactly what we need. In terms of $h$ alone, the same form involves the inverse of the feedback in \eqref{eq:reformed} and has no such simple description.

To control the cross terms in \eqref{eq:again-rewritten}, we use the covariance lemma of \cite{Ban24} to ensure that $\Cov(F) \preceq O(1) \cdot \diag(\Cov(F))$ and $\Cov(G) \preceq O(1) \cdot \diag(\Cov(G))$, which allows us to upper bound the cross terms by the diagonal terms. We will further ensure that $F$ and $G$ are supported on subspaces where the feedback terms $AF$ and $BG$ in \eqref{eq:reformed} don't contribute much; this allows us to further bound the diagonal terms of \eqref{eq:again-rewritten} in terms of the diagonal of $H$ (and hence $h$).

We give the details in the next lemma, which is the most crucial step of the proof. It gives a random direction whose expected energy is at most a constant times the per-coordinate weights $\widetilde{a}_i\widetilde{p}_i + \widetilde{b}_i\widetilde{q}_i$ and $\E[h_i^2]$, and which still moves the coordinates by a nonnegligible amount.
\begin{lemma}[Covariance]\label{lem:covariance-bound}
    There exists $1 \leq C_1 \leq \mathsf{poly}(N, d)$ independent of the current state, such that there is a centered, finitely supported random direction $h \in \R^n$ satisfying, for $\omega_i = T_i^2\E [h_i^2]$,
    \begin{equation}\label{eq13}
    \E E_h \leq 8\sum_{i \in [n]} (\widetilde{a}_i \widetilde{p}_i + \widetilde{b}_i \widetilde{q}_i)\omega_i, \qquad
    \sum_{i \in [n]} \omega_i \geq \frac{n}{C_1}, \qquad \omega_i \leq C_1.
    \end{equation}
\end{lemma}
\begin{proof}
    By the discussion above, we will choose $(F, G)$ at random with a covariance dominated by its diagonal, so that the only relevant terms are the diagonal coefficients $\widetilde{a}_i\widetilde{p}_i$, $\widetilde{b}_i\widetilde{q}_i$ of \eqref{eq:again-rewritten}; following this we may then define $H$ by \eqref{eq:reformed} and set $h_i = H_i/T_i$ so that $E_h$ is given by \eqref{eq:again-rewritten}. We'll restrict $(F, G)$ to a subspace on which the feedback terms $BG$ and $AF$ in \eqref{eq:reformed} are small compared with $F$ and $G$ and act as a contraction, and finally show that $H$ cannot be much smaller than $(F, G)$.

    We need two properties of $A, B$. First, observe that the dual equations give
    \begin{equation}\label{eq15}
    (A\widetilde{q})_i = \widetilde{p}_i - \eps \ell_i v_i^*X^2 v_i,
    \qquad
    (B\widetilde{p})_i = \widetilde{q}_i - \eps m_i v_i^*Y^2 v_i,
    \end{equation}
    as, indeed,
    \[
    (A\widetilde{q})_i = c\ell_i\sum_{j \in [n]} \ell_j m_j q_j|v_i^*Xv_j|^2 = \ell_i v_i^* X \eta_x(Q) X v_i = \ell_i v_i^* (P - \eps X^2) v_i,
    \]
    and the identity for $B\widetilde{p}$ follows by interchanging the two blocks. In particular, $A\widetilde{q} \leq \widetilde{p}$ and $B\widetilde{p} \leq \widetilde{q}$, where these inequalities are entrywise and $\widetilde{p},\widetilde{q}$ are strictly positive. Note further that \Cref{lem4.2} gives $A_{ii} = c\widetilde{a}_i^2 = O(1/c)$ and $B_{ii} = c\widetilde{b}_i^2 = O(1/c)$. (From \eqref{eq6}, a coordinate with $A_{ii}$ or $B_{ii}$ of order one would have a reservoir large enough to pay for its endpoint move.) Because $A, B \succeq 0$, we have that $\|A\|, \|B\| = O(N/c)$.
    
    Recall from \eqref{eq:again-rewritten} that $E_h$ is a positive semidefinite quadratic form in $(F, G)$ with coefficients $\widetilde{a}_i\widetilde{p}_i$ on $F_i^2$ and $\widetilde{b}_i\widetilde{q}_i$ on $G_i^2$. We therefore use the diagonal weights
    \[
    D_F = \diag(\widetilde{a}_i \widetilde{p}_i), \qquad
    D_G = \diag(\widetilde{b}_i \widetilde{q}_i), \qquad
    \|u\|_D^2 = u^\top D u.
    \]
    We will restrict $F, G$ to a large subspace on which the two interactions in \eqref{eq:reformed} are contractions in these weighted norms. To show that only a few directions must be removed, use $|v_i^* X v_j|^2 \leq a_ia_j$ to obtain $A_{ij}^2 \leq c\widetilde{a}_i\widetilde{a}_j A_{ij}$; it follows then that
    \begin{align*}
        \left\|D_G^{1/2}AD_F^{-1/2}\right\|_F^2 &= \sum_{i, j \in [n]} \frac{\widetilde{b}_i\widetilde{q}_i}{\widetilde{a}_j\widetilde{p}_j} A_{ij}^2
        \leq c\sum_{j \in [n]} \frac{1}{\widetilde{p}_j}\sum_{i \in [n]} \widetilde{a}_i\widetilde{b}_i\widetilde{q}_i A_{ij} < \frac{1}{c}\sum_{j \in [n]} \frac{(A\widetilde{q})_j}{\widetilde{p}_j} \leq \frac{n}{c}.
    \end{align*}
    Here the second inequality uses $\widetilde{a}_i\widetilde{b}_i < 1/c^2$ from Lemma \ref{lem4.2} and the last uses $A\widetilde{q} \leq \widetilde{p}$ from \eqref{eq15}. Interchanging the two blocks gives $\|D_F^{1/2} B D_G^{-1/2}\|_F^2 < n/c$.

    Now fix a contraction threshold\footnote{The contraction threshold $\beta_*$ can be understood as the fraction of a response that we allow to come back as feedback through the reservoirs; smaller $\beta_*$ makes the response closer to the direct path but limits the number of directions in which to move.
    } $\beta_* = 3/8$. The squared singular values of a matrix sum up to its squared Frobenius norm, and thus each weighted map has feweer than $n/(c\beta_*^2)$ singular values greater than $\beta_*$. Restrict $\smash{D_F^{1/2}F}$ to the right singular subspace of $\smash{D_G^{1/2}AD_F^{-1/2}}$ with singular values at most $\beta_*$, and impose the analogous restriction on $\smash{D_G^{1/2}G}$. These are linear restrictions of total codimension less than $2n/(c\beta_*^2)$, and they ensure
    \begin{equation}\label{eq16}
    \|AF\|_{D_G} \leq \beta_*\|F\|_{D_F},
    \qquad
    \|BG\|_{D_F} \leq \beta_*\|G\|_{D_G}.
    \end{equation}
    For \eqref{eq:reformed} to define a common vector $H$, we also require $F - BG = AF - G$, or, equivalently, $(I-A)F + (I-B)G = 0$; this imposes at most $n$ additional linear constraints. Let $\mathcal{G} \subseteq \R^{2n}$ be the subspace satisfying all these restrictions. Then $\dim \mathcal{G} \geq n - 2n/(c\beta_*^2) = \frac{29}{45}n$. For every $(F, G) \in \mathcal{G}$, define $H$ by \eqref{eq:reformed}. Those equations can be rewritten as
    \[
    (F, G) = (H, -H) + (BG, AF).
    \]
    On $\mathcal{G}$ the feedback term $(BG, AF)$ is at most a $\beta_*$ fraction of the size of $(F, G)$, so $(F, G)$ is comparable to the direct term $(H, -H)$.
    Using the product norm determined by $D_F, D_G$, and then \eqref{eq16}, gives
    \[
    \sqrt{\|F\|_{D_F}^2 + \|G\|_{D_G}^2}
    \leq \sqrt{\sum_{i \in [n]} (\widetilde{a}_i \widetilde{p}_i + \widetilde{b}_i\widetilde{q}_i)H_i^2} + \beta_*\sqrt{\|F\|_{D_F}^2 + \|G\|_{D_G}^2}.
    \]
    Moving the last term to the left and squaring both sides then yields
    \begin{equation}\label{eq17}
        \|F\|_{D_F}^2 + \|G\|_{D_G}^2 \leq \frac{1}{(1-\beta_*)^2} \sum_{i \in [n]} (\widetilde{a}_i \widetilde{p}_i + \widetilde{b}_i\widetilde{q}_i)H_i^2 = \frac{64}{25}\sum_{i \in [n]} (\widetilde{a}_i \widetilde{p}_i + \widetilde{b}_i\widetilde{q}_i)H_i^2.
    \end{equation}
    We now choose a random pair in $\mathcal{G}$. Apply \Cref{lem:covariance} with $\kappa=25/8$; in ambient dimension $2n$, it gives $\Gamma \succeq 0$ satisfying $\mathrm{range}(\Gamma) \subseteq \mathcal{G}$, $\Gamma_{jj} \leq 1$, $\Gamma \preceq \kappa\diag(\Gamma)$, and
    \begin{equation}\label{eq18}
        \operatorname{Tr}{\Gamma} \geq \dim{\mathcal{G}} - \frac{2n}{\kappa} \geq \biggl(1 - \frac{2}{c\beta_*^2}-\frac{2}{\kappa}\biggr) = \frac{n}{225}.
    \end{equation}
    where the last inequality holds for sufficiently large absolute $c, \kappa$.

    Take $(F, G) = \Gamma^{1/2}\xi$, where $\xi \in \{\pm 1\}^{2n}$ has independent uniform coordinates. This pair is centered, finitely supported, has covariance $\Gamma$, and lies in $\mathcal{G}$. Define $H$ by \eqref{eq:reformed} and $h_i = H_i/T_i$. Then $h$ is also centered and finitely supported, and $\omega_i = \E H_i^2$.

    Now to verify the energy guarantee, let $\mathsf{M} \succeq 0$ be the matrix of the quadratic form $\operatorname{Tr}(PC_FXC_F) + \operatorname{Tr}(QC_GYC_G)$. Its diagonal is $(D_F, D_G)$, as computed above. Since the covariance of $(F, G)$ is $\Gamma$,
    \begin{align*}
    \E E_h = \operatorname{Tr}(\mathsf{M}\Gamma) &\leq \kappa \operatorname{Tr}(\mathsf{M}\diag(\Gamma))
    = \kappa \E\bigl[\|F\|_{D_F}^2 + \|G\|_{D_G}^2\bigr]\\ &\leq \frac{25}{8}\cdot\frac{64}{25}\sum_{i \in [n]} (\widetilde{a}_i\widetilde{p}_i + \widetilde{b}_i\widetilde{q}_i)\omega_i
    = 8\sum_{i \in [n]} (\widetilde{a}_i\widetilde{p}_i + \widetilde{b}_i\widetilde{q}_i)\omega_i.
    \end{align*}
    Here the first inequality follows from $\mathsf{M} \succeq 0$ and $\Gamma \preceq \kappa\diag(\Gamma)$; the last inequality is \eqref{eq17}.\footnote{Note that this where $\Gamma \preceq \kappa\,\diag(\Gamma)$ is crucially useful; it lets us discard the off-diagonal entries of $\mathsf{M}$, which consist of the interactions between different coordinates in $E_h$.}

    It remains then to establish the movement bounds. \eqref{eq18} controls $\E[\|F\|^2 + \|G\|^2]$ but does not yet control $\E \|H\|^2$; the two terms defining $H$ in \eqref{eq:reformed} could nearly cancel. We rule out excessive cancellation using the regularization terms in \eqref{eq15}. The normalization gives $\|S\| \leq 2$, $\nu \leq 1$, and $d\eps = \nu$, so the upper bound from \Cref{lem:matrix_potential} implies that
    \[
    R(x) \leq \|S\| + 2\sqrt{c\nu + 2d\eps} \leq 2 + 2 \sqrt{(c+2)\nu} = O(\sqrt{c}).
    \]
    This upper bound follows by testing scalar matrices in the defining optimization. Since $t \leq R(x)$, the constraints give $X^{-1}, Y^{-1} \preceq (t + \|S\|)I \preceq O(\sqrt{c})I$, and hence $X, Y \succeq \Omega(c^{-1/2})I$. Also, \eqref{eq7} implies $P, Q \preceq I$, and thus for a sufficiently small absolute constant $\gamma_0 > 0$,
    \[
    \eps\ell_iv_i^*X^2v_i \geq \frac{\gamma_0\eps}{c}\ell_i\|v_i\|^2 \geq \frac{\gamma_0 \eps}{c}\widetilde{p}_i,
    \]
    and the analogous inequality holds for the $Y$-term in \eqref{eq15}. Consequently,
    \begin{equation}\label{eq19}
        A\widetilde{q} \leq (1-s_0)\widetilde{p},\qquad B\widetilde{p} \leq (1-s_0)\widetilde{q}, \qquad s_0 = \gamma_0\eps/c \in (0, 1).
    \end{equation}
    Since $A$ is entrywise nonnegative, \eqref{eq19} gives $AB\widetilde{p} \leq (1-s_0)^2\widetilde{p}$. Conjugating $AB$ by the positive diagonal matrix $\diag(\widetilde{p})$ produces a nonnegative matrix whose row sums are at most $(1-s_0)^2$; its spectral radius, and therefore that of $AB$, is at most $(1-s_0)^2$.

    The matrix $A^{1/2}B A^{1/2}$ is positive semidefinite and has the same nonzero eigenvalues as $AB$, hence
    \[
    \bigl\|(I - A^{1/2}BA^{1/2})^{-1}\bigr\| \leq
    \frac{1}{1-(1-s_0)^2} = O(1/s_0).
    \]
    Using the identity $(I - AB)^{-1} = I + A^{1/2}(I - A^{1/2}BA^{1/2})^{-1}A^{1/2}B$, we obtain
    \begin{align*}
        \bigl\|(I - AB)^{-1}\bigr\| &\leq 1 + O\biggl(\frac{\|A\|\|B\|}{s_0}\biggr) = O\biggl(1 + \frac{N^2}{c\eps}\biggr)
        \leq O(1 + N^2/\eps),
    \end{align*}
    and the same bound holds with $A, B$ interchanged. 
    
    We can now solve \eqref{eq:reformed} quantitatively. Substituting  $G = AF - H$ into $F = H + BG$, and conversely substituting $F = H + BG$ into $G = AF - H$, gives
    \[
    F = (I-BA)^{-1}(I-B)H, \qquad G = (I-AB)^{-1}(A-I)H.
    \]
    In particular, \eqref{eq:reformed} has a unique solution $(F, G)$ for each $H$. The inverse bounds and $\|A\|, \|B\| = O(N/c)$ imply $\|F\| + \|G\| \leq O((1+N)^3\eps^{-1})\|H\|$, where we used $0 < \eps \leq 1$. Squaring and taking expectations then yields
    \[
    \operatorname{Tr}{\Gamma} = \E[\|F\|^2 + \|G\|^2] \leq O((1+N)^6\eps^{-2})\E\|H\|^2
    \implies \sum_{i \in [n]} \omega_i = \E\|H\|^2 \geq \frac{n}{C_1},
    \]
    since $\operatorname{Tr}{\Gamma} \geq n/225$. Finally, $\Gamma \preceq \kappa I$ since $\Gamma \preceq \kappa \diag(\Gamma) \preceq \kappa I$, and $H = [I, -B](F, G)$. Therefore
    \begin{align*}
        \omega_i &= e_i^\top [I, -B]\Gamma[I, -B]^\top e_i \leq \kappa e_i^\top (I + B^2)e_i\leq \kappa(1 + \|B\|^2) = O((1+N)^2) \leq C_1,
    \end{align*}
    and so we have established all three guarantees in \eqref{eq13}.
\end{proof}
\begin{remark}
    The parameters $\beta_*$ and $\kappa$ in the proof of Lemma \ref{lem:covariance-bound} above trade the size of the energy against the room left to move. Lowering either one improves the energy constant $\kappa/(1-\beta_*)^2$, but the contraction requirement then removes $2n/(c\beta_*^2)$ directions and \Cref{lem:covariance} gives up another $2n/\kappa$, and together these must stay below $n$. We take $\beta_* = 3/8$ and $\kappa = 25/8$, which nearly exhausts this budget, leaving $\Tr\Gamma \ge n/225$, and in exchange brings the energy constant down to $8$, just within what \eqref{eq21} can afford at $c = 40$.
\end{remark}
We now have enough machinery to prove \Cref{lem:potential_decrease}.
\begin{proof}[Proof of \Cref{lem:potential_decrease}]
    The endpoint alternative was established above, so assume \eqref{eq6}. Let $h$ be the random direction from \Cref{lem:covariance-bound}, and set $\Sigma = \E[hh^\top]$. Then $\Sigma \succeq 0$, and for every real symmetric matrix $M$, we have $\E[h^\top M h] = \operatorname{Tr}(M\Sigma)$. For each $i$, write
    \begin{align*} \label{eq:coordinate-weights}
    \rho_i = \widetilde{a}_i\widetilde{p}_i + \widetilde{b}_i\widetilde{q}_i, \qquad \rho_i^\times = \widetilde{a}_i\widetilde{q}_i + \widetilde{b}_i\widetilde{p}_i
    \end{align*}
    for the direct weight and the crossed weight respectively. We first combine the bound on $\nabla^2 R$ in \Cref{lem:derivatives-R} with the covariance guarantee of \Cref{lem:covariance-bound}. The rescaling identities \eqref{eq:scales} give
    \[
    \psi_i(b_ip_i + a_iq_i) = \widetilde{b}_i\widetilde{p}_i + \widetilde{a}_i\widetilde{q}_i = \rho_i^\times, \qquad \E h_i^2 = \frac{\omega_i}{T_i^2} = \frac{\psi_i \omega_i}{\alpha_i \beta_i}.
    \]
    Consequently,
    \[
    \E\sum_{i \in [n]} (\chi_i h_i^2(b_ip_i + a_iq_i) = \sum_{i \in [n]} \frac{\chi_i}{\alpha_i \beta_i}\rho_i^\times\omega_i \geq \frac{5}{9}\sum_{i \in [n]} \rho_i^\times\omega_i,
    \]
    where the last inequality follows from \Cref{lem4.2}. Applying \Cref{lem:derivatives-R} and then the energy bound in \Cref{lem:covariance-bound}, we obtain the bound 
    \begin{equation}\label{eq21}
    \frac{1}{2}\operatorname{Tr}((\nabla^2 R)\Sigma) \leq 11\sum_{i \in [n]}\rho_i\omega_i - \frac{5c}{18} \sum_{i \in [n]} \rho_i^\times\omega_i.
    \end{equation}
    The negative sum carries the factor $c$, but its weights are crossed, and for a given $i$ the weights $\rho_i$ and $\rho_i^\times$ need not be comparable. Their difference factors as 
    \begin{equation}\label{eq:straight-crossed}
    \rho_i-\rho_i^\times = (\widetilde{a}_i - \widetilde{b}_i)(\widetilde{p}_i - \widetilde{q}_i).
    \end{equation}
    By \Cref{lem:derivatives-R}, $\partial_i R = \alpha_i p_i - \beta_i q_i = T_i(\widetilde{p}_i - \widetilde{q}_i)$, so the difference is proportional to the slope of $R$ in coordinate $i$, and a drift can pay for it. Choose
    \begin{equation}\label{eq22}
        \mu_i = -\frac{11}{T_i}(\widetilde{a}_i - \widetilde{b}_i) \omega_i.
    \end{equation}
    Its contribution to $R$ is 
    \[
    \nabla R \cdot \mu = -11 \sum_{i \in [n]} (\widetilde{a}_i - \widetilde{b}_i)(\widetilde{p}_i - \widetilde{q}_i)\omega_i = -11\sum_{i \in [n]}(\rho_i - \rho_i^\times)\omega_i,
    \]
    and adding this identity to \eqref{eq21} replaces the positive sum by the crossed weights appearing in the negative sum.

    For the rest of the proof, write
    \[
    \mathcal{D}[f] = \nabla f \cdot \mu + \frac{1}{2}\operatorname{Tr}(\nabla^2f\, \Sigma).
    \]
    The preceding cancellation gives, since $c=40$,
    \[
    \mathcal{D}[R] \leq \biggl(11 - \frac{5c}{18}\biggr)\sum_{i \in [n]} \rho_i^\times\omega_i = -\frac{1}{9}\sum_{i\in [n]}\rho_i^\times \omega_i\leq 0.
    \]
    To obtain a uniform amount of decrease, we now use $\Phi$; we must first check that the drift in \eqref{eq22} does not consume the decrease supplied by its negative second derivatives. First, the rescaling and the definitions of $\alpha_i, \beta_i$ give the cancellation
    \[
    T_i(\widetilde{a}_i + \widetilde{b}_i) = \alpha_i a_i + \beta_i b_i = (1+c\psi_i'b_i)a_i + (1-c\psi_i'a_i)b_i = a_i + b_i.
    \]
    It follows from \eqref{eq22} that
    $
    |\mu_i| \leq {11}(\widetilde{a}_i + \widetilde{b}_i)\omega_i/{T_i} = 11(a_i + b_i)\omega_i/{T_i^2}.
    $
    Adding the complementary inequalities in \eqref{eq6} gives $a_i + b_i < 2/(c\psi_i)$, and using this with \Cref{lem4.2} gives
    \[
    |\psi_i' \mu_i| \leq \frac{22}{c} \frac{|\psi_i'|}{\psi_i}\frac{\omega_i}{T_i^2} \leq \frac{11}{c}(-\psi_i'')\frac{\omega_i}{T_i^2} = \frac{11}{40}(-\psi_i'')\frac{\omega_i}{T_i^2},
    \]
    and thus the drift costs only an $\frac{11}{20}$-fraction of the available coordinate variance $\frac{1}{2}(-\psi_i'')\omega_i/T_i^2$. Because $\Phi = \sum_i \psi_i$, its Hessian is diagonal; moreover, $\Sigma_{ii} = \omega_i/T_i^2$, and therefore
    \begin{align*}
        \mathcal{D}[\Phi] &= \sum_{i \in [n]} \psi_i'\mu_i - \frac{1}{2}\sum_{i \in [n]} (-\psi_i'')\frac{\omega_i}{T_i^2} \leq -\biggl(\frac{1}{2} - \frac{11}{40}\biggr)\sum_{i \in [n]} (-\psi_i'')\frac{\omega_i}{T_i^2}\\
        &\leq -\Omega(1)\sum_{i \in [n]} \omega_i \leq -\Omega(n/C_1).
    \end{align*}
    In the above, the first inequality uses the bound on $|\psi_i' \mu_i|$ above; the second uses $(-\psi_i'')/T_i^2 = \Omega(1)$ from Lemma \ref{lem4.2}; and the third uses the movement bound in Lemma \ref{lem:covariance-bound}.) It remains then to verify the size bounds on $\mu, \Sigma$. The same cancellation gives
    $
    (\widetilde{a}_i + \widetilde{b}_i)/T_i =
    (\psi_i(a_i + b_i))/(\alpha_i \beta_i) = O(1/c),
    $
    using \eqref{eq6} and $\alpha_i\beta_i = \Omega(1)$. Together with $1/T_i^2 = O(1)$ and $\omega_i \leq C_1$, this implies
    \[
    \|\mu\|_\infty = O(C_1), \qquad
    \operatorname{Tr}{\Sigma} = \sum_{i \in [n]}\frac{\omega_i}{T_i^2} = O(C_1n).
    \]
    Finally, $\mathcal{D}[\Psi] = \mathcal{D}[R] + \lambda\mathcal{D}[\Phi] \leq -\Omega(\lambda n/C_1)$, and taking $C_0$ to be a sufficiently large absolute multiple of $C_1$ proves all three bounds in the second alternative.

    We now conclude by showing that this certificate corresponds to a local decrease in expectation. Consider $\Delta = \sqrt{\theta} h + \theta \mu$. Because $h$ has finite support and the current state lies in an open face, all points $x + \Delta$ remain in that face for sufficiently small $\theta > 0$. Since $\E h = 0$, we have $\E\Delta = \theta \mu$ and $\E[\Delta\Delta^\top] = \theta \Sigma + \theta^2 \mu\mu^\top$. The Taylor expansion in a fixed neighbourhood of $x$ therefore gives
    \begin{align*}
        \E\Psi(x + \sqrt{\theta}h + \theta \mu) &= \Psi(x) + \theta\left(\nabla \Psi \cdot \mu + \frac{1}{2}\operatorname{Tr}(\nabla^2 \Psi\, \Sigma)\right) + O_x(\theta^{3/2})\\
        &= \Psi(x) + \theta\,\mathcal{D}[\Psi] + O_x(\theta^{3/2}).
    \end{align*}
    The coefficient of $\theta$ is strictly negative, so the above is smaller than $\Psi(x)$ for sufficiently small $\theta$.
\end{proof}

\subsection*{Acknowledgements}
We thank Lorenzo Orecchia, Clayton Mizgerd, and Nikhil Bansal for helpful discussions.
\paragraph{Statement on AI Usage.} The proof in this article was found by Anthropic's Claude Fable 5.1 model, working from material supplied by the authors. The authors provided Claude with (1) notes from earlier attempts at a potential-function proof of Weaver's $\KS_2$ problem, (2) notes from other projects, including a long-running project on finite-free probability, and (3) a prompt describing at a high level how a potential-function argument should proceed followed by several rounds of interaction and guidance; Claude subsequently translated the language used across fields and produced the proof presented here which was subsequently polished by the authors. The main ideas are connected to the supplied material: the choice of the matrix potential $R(x)$ with the coupled reservoir is  related to notions in free probability,\footnote{
The potential can be understood as a matrix-valued analogue of the inverse-Cauchy-transform barriers for finite-free convolution \cite{marcus2022finite}. Specifically, $R(x)$ equals the \emph{upper spectral edge}---namely the largest element of the spectrum, generalizing the largest eigenvalue of a Hermitian matrix---of an auxiliary operator with mean $\operatorname{diag}(S, -S)$ and off-diagonal free semicircular noise. By Lehner's variational formula \cite{lehner1999computing, parmaksiz2025computing} we may then obtain the coupled $X, Y$ optimization, and the trace penalty corresponds to adding isotropic semicircular noise.
} and the movements used to decrease the potential, endpoint movement and random local steps, had been considered in the earlier attempts on $\KS_2$. The authors remain fully responsible for the content of the article.

\bibliographystyle{alphaurl}
\bibliography{bib.bib}

\end{document}